\documentclass[num-refs]{wiley-article}

\usepackage{diffcoeff}
\usepackage{fullpage,amsmath,tkz-base}  
\usepackage{cases}
\newcommand{\floor}[1]{\left \lfloor #1 \right \rfloor}

\usepackage{rustic} 
\usepackage[T1]{fontenc}

\usepackage{duerer} 
\usepackage[T1]{fontenc}

\usepackage{trajan}
\usepackage[T1]{fontenc}

\usepackage{wrapfig}
\usepackage{upgreek}

\usepackage{siunitx}
\usepackage{framed}

\usepackage{stackengine}

\usepackage{autobreak}
\allowdisplaybreaks

\def \B{\mathcal{B}}
\def \C{\mathcal{C}}

\def \F{\mathcal{F}}

\def \W{\mathcal{W}}
\def \X{\mathcal{X}}

\def \bW {\mathcal{W}}

\def \sT {\mathscr{T}}

\def \fc{\mathbf{c}}
\def \fc{\mathbf{c}}
\def \fC{\mathbf{C}}

\def \fx{\mathbf{x}}
\def \fy{\mathbf{y}}

\def \fY{\mathbf{Y}}

\def \f0{\mathbf{0}}

\definecolor{blau_1a}{RGB}{93,133,195}
\definecolor{blau_2a}{RGB}{0,156,218}
\definecolor{gruen_3a}{RGB}{80,182,149}
\definecolor{gruen_4a}{RGB}{175,204,80}
\definecolor{gruen_5a}{RGB}{221,223,72}
\definecolor{orange_6a}{RGB}{255,224,92}
\definecolor{orange_7a}{RGB}{248,186,60}
\definecolor{rot_8a}{RGB}{238,122,52}
\definecolor{rot_9a}{RGB}{233,80,62}
\definecolor{lila_10a}{RGB}{201,48,142}
\definecolor{lila_11a}{RGB}{128,69,151}

\definecolor{blau_1b}{RGB}{0,90,169}
\definecolor{blau_2b}{RGB}{0,131,204}
\definecolor{gruen_3b}{RGB}{0,157,129}
\definecolor{gruen_4b}{RGB}{153,192,0}
\definecolor{gruen_5b}{RGB}{201,212,0}
\definecolor{orange_6b}{RGB}{253,202,0}
\definecolor{orange_7b}{RGB}{245,163,0}
\definecolor{rot_8b}{RGB}{236,101,0}
\definecolor{rot_9b}{RGB}{230,0,26}
\definecolor{lila_10b}{RGB}{166,0,132}
\definecolor{lila_11b}{RGB}{114,16,133}

\definecolor{mycolor1}{rgb}{0.0, 0.18, 0.39}
\definecolor{mycolor2}{RGB}{87,108,67}
\definecolor{mycolor3}{RGB}{8,133,161}
\definecolor{mycolor4}{RGB}{80,91,161}
\definecolor{mycolor5}{RGB}{98,122,157}
\definecolor{mycolor6}{RGB}{255,163,67}
\definecolor{mycolor7}{RGB}{152,205,225}
\definecolor{mycolor8}{RGB}{242,204,48}
\definecolor{mycolor9}{rgb}{0,.5,0}
\definecolor{mycolor10}{rgb}{.59,.44,.09}
\definecolor{mycolor11}{RGB}{231,199,31} 
\definecolor{mycolor12}{RGB}{8,133,161} 
\definecolor{mycolor13}{RGB}{157,188,64} 
\definecolor{mycolor14}{RGB}{194,150,130} 
\definecolor{mycolor15}{RGB}{98,122,157} 
\definecolor{mycolor16}{RGB}{160,160,160} 
\definecolor{mycolor17}{RGB}{115,82,68} 
\definecolor{mycolor18}{RGB}{94,60,108} 
\definecolor{mycolor19}{RGB}{115,82,68} 
\definecolor{mycolor20}{RGB}{255,183,30} 
\definecolor{mycolor21}{RGB}{185,146,103} 

\usepackage{commath,amsfonts}
  \usepackage{amsmath}
   \usepackage{amssymb}
   \usepackage{bm}

\usepackage{footmisc}

\usepackage[justification=justified]{caption}
\DeclareCaptionFont{jkpl}{\fontsize{8}{11}\color{gray!70!black}\fontfamily{jkpl}\selectfont}
\usepackage{mathtools}
\usepackage{algorithmic}
\usepackage{pgfplots} 
\usetikzlibrary{arrows}
\pgfplotsset{compat=1.15}
\usepackage{graphicx}
\usepackage{xfrac}
\usepackage{colortbl}
\usepackage{cancel}
\usepgfplotslibrary{fillbetween}
\usepackage{float}
\usepackage{hyperref}
\usepackage{multirow}
\usepackage{xcolor}
\usepackage{mathrsfs}
\usepackage{bbm}
\usepackage{autobreak}
\allowdisplaybreaks
\usepackage{tikz}
\usetikzlibrary{arrows.meta}
\usepackage{amsbsy}
\usepackage{fixmath}
\usepackage{silence}
\usetikzlibrary{math}
\usepackage{graphicx}
\usepackage{pgfgantt}
\usepackage{pdflscape}

\usepackage{upgreek}
\usepackage{nicematrix}

\usepackage{leftindex}
\let\svtikzpicture\tikzpicture
\def\tikzpicture{\noindent\svtikzpicture}

\DeclareMathOperator*{\argmax}{arg\,max}

\def \bG {\mathcal{G}_{\rm slow}}

\papertype{ArXiv Preprint}
\paperfield{Information Theory}

\title{{\normalfont\duinfamily\bfseries DETERMINISTIC K$-$IDENTIFICATION \\ FOR BINARY SYMMETRIC CHANNEL}}

 \abbrevs{\fontfamily{jkpl}\selectfont \textcolor{gray!70!black}{Ons Dabbabi is with the Institute for Communications Engineering (ICE), Technical University of Munich (TUM), Germany (E-Mail: \{ge34xov\}@tum.de). Mohammad Javad Salariseddigh and Christian Deppe are with the ICE, German Federal Ministry of Education and Research (BMBF); Research Hub 6G-Life, TUM, Germany (E-Mail: \{mjss,~christian.deppe\}@tum.de). Holger Boche is with the Chair of Theoretical Information Technology; BMBF Research Hub 6G-Life, TUM, Germany (E-Mail: boche@tum.de).}}

  \author{{\small\durmfamily ONS DABBABI}\,, {\small\durmfamily MOHAMMAD JAVAD SALARISEDDIGH}\,, \\{\small\durmfamily CHRISTIAN DEPPE}\,, \text{\small\normalfont and } {\small\durmfamily HOLGER BOCHE}}

\begin{document}

\begin{frontmatter}
\maketitle

\begin{abstract}
Deterministic K-Identification (DKI) for the binary symmetric channel (BSC) is developed. A full characterization of the DKI capacity for such a channel, with and without the Hamming weight constraint, is established. As a key finding, we find that for deterministic encoding the number of identifiable messages $K$ may grow exponentially with the codeword length $n$, i.e., $K = 2^{\kappa n}$, where $\kappa$ is the target identification rate. Furthermore, the eligible region for $\kappa$ as a function of the channel statistics, i.e., the crossover probability, is determined.
\keywords{Deterministic K-Identification, Binary Symmetric Channel, Hamming Distance}
\end{abstract}
\end{frontmatter}
\section{Introduction}
\fontfamily{jkpl}\selectfont
Several applications in the context of post Shannon communications \cite{Schwenteck23,6G+,6G_PST} for the future-generation wireless networks (XG) horizon are based on event-triggered communication settings. In such systems, Shannon's message transmission capacity, as investigated early in \cite{S48}, may not be the suitable performance measure, instead, the identification capacity is regarded to be the key applicable quantitative metric. Specifically, for event-recognition, alarm-prompt or object-finding problems, where the receiver aims to recognize the occurrence of a specific event, determine an alarm, or realize the presence of an object in terms of a \emph{reliable} Yes\,/\,No final decision, the so-called identification capacity is the key applicable performance measure \cite{AD89}.

While in Shannon's communication paradigm \cite{S48}, sender, encodes its message in a manner that the receiver can perform a reliable \emph{reproduction}, in the identification setting \cite{AD89}, the coding scheme is designed to accomplish a different objective, namely, to determine whether a \emph{particular} message was sent or not. The identification problem in communication theory is initiated\footnote{\pagecolor{gray!3!yellow!3}\,\textcolor{gray!70!black}{The identification problem has been studied in various setting of deterministic or randomized protocols, in the context of communication complexity; see \cite{Abu88,Yao79}}} by Ahlswede and Dueck \cite{AD89} where a randomized encoder, is employed, to select the codewords. Therein, codebook consists of distributions and its size grows double-exponentially in the codeword length $n$, i.e., $\sim 2^{ 2^{nR}}$ \cite{AD89}, where $R$ is the coding rate. The realization of explicitly constructed randomized identification (RI) codes features high complexity and is challenging for the applications; cf. \cite{Salariseddigh22} for further details.
The motivation of Ahlswede and Dueck to develop the RI problem \cite{AD89} is probably traced back to the work of J\'aJ\'a \cite{J85} who considered deterministic identification (DI) from a communication complexity\footnote{\,\textcolor{gray!70!black}{An important observation regarding the behavior of the identification function has been well studied in communication complexity where the out-performance of randomized protocols over the deterministic protocols (exponential gap between the two class) for computing such function is established. For instance, while the error-free deterministic complexity of the identification function is lower bounded by $\log m$, where $m$ is the length of message, for the randomized protocol and when $\varepsilon$ error is allowed in computation of the identification function, only $\mathcal{O}(\log \log m + \frac{1}{\varepsilon})$ bits suffices; see \cite{Yao79,Mehlhorn82} for further details.}} perspective, that is, where the codewords are determined by a deterministic function from the messages. Further, it seems that Ahlswede and Dueck were inspired to show that employing randomness similar to what has been accomplished in the communication complexity field, yield an advantage of exponential gap over the DI problem\footnote{\,\textcolor{gray!70!black}{A detailed comparison of codebook sizes in DI and RI problem over various channel models can be found in \cite{Salariseddigh22}}} in terms of the codebook size. DI may be favored over RI in complexity-constrained applications. For instance, in molecular communication (MC) systems where development and deploying of huge number of random number generators may not be clear.

The binary symmetric channel (BSC) is deemed as a basic mathematical model through which one bit per unit of time can be transmitted. The capacity of such a channel is attained by Bernoulli input with $1/2$ success probability, i.e., $X \sim \text{Bern}(1/2)$. In \cite{Elias55} consider the BSC and used a random linear code for the achievability proof which result in a exponential search in the decoding. 

The DI problem for discrete memoryless channel (DMC) where codewords are restricted by an average power constraint, is studied in \cite{Salariseddigh_ICC,Salariseddigh_IT}. Therein, employing the \emph{method of types}\footnote{\,\textcolor{gray!70!black}{The method of type developed and promoted by Csiszar and Körner and treated in depth in \cite{CK82,Csiszar98}. Such a method is regarded as a widely used and fundamental technique to obtain capacity results for different source and channel coding settings within the context of mathematical information theory; see \cite{Kramer08} for further details with related topics on multi-user models. A survey of recent developments with applications in statistics can found in \cite{Csiszar04}. Further examples and in-depth mathematical details with applications to large deviation theory can be found in \cite[Ch. 13]{Bremaud17}.}} and standard techniques, it is established that the codebook size grows exponentially as a function of the codeword length, i.e., $\sim 2^{nR}$ \cite{Salariseddigh_ICC,Salariseddigh_IT}. This observation acknowledge that the codebook size of DI over DMCs behaves similar to that of the message transmission problem \cite{S48}. DI for the compound channels is studied by Ahlswede and Cai in \cite{AN99}. Furthermore, the DI for continuous alphabet channels including Gaussian channels with fast and slow fading, memoryless and inter-symbol interference (ISI)-aware discrete-time Poisson channel (DTPC) is addressed in \cite{Salariseddigh_ITW,Salariseddigh_arXiv_ITW,Salariseddigh_IT,Salariseddigh_GC_IEEE,Salariseddigh_GC_arXiv,Salariseddigh22_2} where the derivation of the lower bound on the DI capacity (achievability proof) has been approached via the technique of \emph{ball packing} \cite{CHSN13} within a subset of the Euclidean space, $\mathbb{R}^n$. In the latter work \cite{Salariseddigh22_2} ISI-dependent bounds on the DI capacity are calculated. For all the continuous alphabet works, i.e., the Gaussian, Poisson (with/out ISI), 
and Binomial models \cite{Salariseddigh_ITW,Salariseddigh_IT,Salariseddigh_GC_IEEE,Salariseddigh_GC_arXiv,Salariseddigh22_2}, a new observation regarding the codebook size is obtained, namely, the codebook size scales \emph{super-exponentially} in the codeword length, i.e., $\sim 2^{(n\log n)R}$ which is different than the standard exponential \cite{Salariseddigh_ICC} and double exponential \cite{AD89} behavior for DI and RI problems, respectively. In \cite{Wiese22} the DI problem with non-discrete additive white noise and noiseless feedback under both average and peak power constraints, is analyzed, where the DI capacity is shown to be infinite regardless of the scaling for the codebook size. The problem of joint identification and channel state estimation for a DMC with independent and identically distributed state sequences, is studied in \cite{Labidi23}. Therein, the \emph{sensing} is viewed as an additional resource that increases the identification capacity.

\subsection{Previous Results}
In the (standard) identification problem \cite{AD89}, the receiver aims to identify the occurrence of a \emph{single} message. However, there exist a generalized variation of the DI problem, called the K-Identification problem \cite{AH08G}, in which the receiver may seek to determine the presence of a single message \emph{within} a set of messages (subset of the message set) referred to as the target message set. The K-Identification problem can be understood as the generalization of the original identification problem in the following fashion: The target message (singleton set) is enlarged to be a general set of $K_{>1}$ messages. Ahlswede in \cite[Th.~1,Propos.~$1$]{AH08G,A80} showed that the number of target messages for RI problem over a DMC scales exponentially in the codeword length $n$, i.e., $K=2^{\kappa n}$ and proved that the set of all deterministic K-identification (DKI) achievable pairs contains
\begin{align}
    \left\{ (R,\kappa) : 0 \leq R, \kappa\; ; \; R+2 \kappa \leq \mathbb{C}_{\rm TR} \right\} \;,
\end{align}
where $\mathbb{C}_{\rm TR}$ is the message transmission capacity of the DMC.

The DKI problem for the slow fading channels $\bG$, assuming that the number of target messages scales sub-linearly with codeword length $n$, i.e., $K(n,\kappa) = 2^{\kappa \log n}$, subject to an average power constraint and a codebook size of super-exponential scale, i.e., $M(n,R)=2^{(n\log n)R}$, is studied in \cite{Salariseddigh22_3} where the following $K$-depending bounds on the DKI capacity are derived:
\begin{align}
    \frac{1-\kappa}{4} \leq \mathbb{C}_{\rm DKI}(\bG,M,K) \leq 1 + \kappa \,.
\end{align}
\subsection{Contributions}
In this paper, we consider identification systems employing deterministic encoder and receivers that are interested to accomplish the K-Identification task, namely, finding an object in a target message set of size $K=2^{\kappa n}$ for $\kappa \in [0,1)$. We assume that the communication over $n$ channel uses are independent of each other. We assume that the noise is additive Bernoulli process and formulate the problem of DKI over the DTBC under Hamming weight input constraint.

To the best of the authors' knowledge, the fundamental performance limits of DKI for the BSC model has not been so far studied in the literature. As our main objective, we investigate the DKI capacity of the BSC. In particular, this paper makes the following contributions:
\begin{itemize}[leftmargin=*]
    \item[$\diamondsuit$] \textbf{\textcolor{mycolor12}{Generalized Identification Model}}: In several identification systems, often the size of target message set $K$ can be large, particularly when one by one comparison is not demanded due to the delay constraint. In addition, the value of $K$ may increases as a function of the codeword lengths $n$. To address these cases, we consider a generalized identification model\footnote{\,\textcolor{gray!70!black}{The proposed \emph{generalized identification} setting may be used in a more advanced scheme called \emph{generalized identification with decoding} \cite[Ch.~1]{Ahlswede06} where first the K-Identification and second, the standard identification are accomplished to find an object. Furthermore, \emph{generalized identification} should be distinguished from \emph{multiple object identification} \cite{Yamamoto15} where $K$ objects whose corresponding target message sets are \emph{unknown} to the receiver, are identified at once.}} that captures the standard (i.e., $K=1$), identification channels with constant $K>1$, and identification channels for which $K$ increases with the codeword length $n$.
    \item[$\diamondsuit$] \textbf{\textcolor{mycolor12}{Codebook Scale}}: We establish that the codebook size of the DKI problem over the BSC for deterministic encoding scales \emph{exponentially} in the codeword length $n$, i.e., $\sim 2^{nR}$, even when the size of target message set scales as $K=2^{\kappa n}$ for some $\kappa \in [0,1)$ (see Theorem~\ref{Th.DKI-Capacity} for exact upper bound on the $\kappa$), which we refer to as the \emph{target identification rate}. Such an exponential scale for the codebook size coincide with that of the message transmission problem \cite{S48} and the standard identification problem (DI) in which $K = 2^0 = 1$ \cite{Salariseddigh_IT,Salariseddigh_ICC,J85} and is lower than the \emph{super-exponential} scale for that of the DKI problem over the slow fading channels \cite{Salariseddigh22_3}. This observation suggests that increasing the number of target messages does not change the scale of the codebook derived for DI scheme over the BSC \cite{Salariseddigh_IT,Salariseddigh_ICC,J85}.
    \item[$\diamondsuit$] \textbf{\textcolor{mycolor12}{Capacity Formula}}: We derived a closed form analytical function for the DKI capacity for the BSC, which are the main results of this paper. Such formula does \emph{reflect} the impact of the input constraint $P_{\,\text{ave}}$ in the \emph{optimal} scale of the codebook size, i.e., $2^{nR}$. This observation is in contrast to the result obtained for the DKI problem for the slow fading channel \cite{Salariseddigh22_3} or the DI problem for Gaussian and Poisson channels \cite{Salariseddigh_ITW,Salariseddigh_IT,Salariseddigh_GC_IEEE,Salariseddigh22_2}. We derive the DKI capacity formula for the BSC with constant $K \geq 1$ and growing size of the target message set $K = 2^{\kappa n}$, respectively. We show that for both cases of the constant $K$ and the growing number of target messages, the proposed capacity expression is not a function of the target identification rate $\kappa$ and remains only as a function of the Hamming weight constraint.
    \item[$\diamondsuit$] \textbf{\textcolor{mycolor12}{Technical Novelty}}: To obtain the proposed lower bound, the existence of an appropriate ball packing within the input space, for which the Hamming distance between the centers of the balls does not fall below a certain value, is established.
    In particular, we consider the packing of hyper balls inside a larger $n$-dimensional Hamming hyper ball, whose radius grows in the codeword length $n$, i.e., $nA$. For the achievability proof, we exploit a greedy construction similar to the \emph{Gilbert bound} method. While the radius of the small balls in the DKI problem for the slow fading channel \cite{Salariseddigh22_3}, grows in the codeword length $n$ as $n \to \infty$, here, the radius similar to the DI problem for the Gaussian channel with slow and fast fading \cite{Salariseddigh_ITW} tends to zero. In general, the derivation of lower bound for the BSC is more involved compared to that for the Gaussian \cite{Salariseddigh_ITW} and Poisson channels with/out memory \cite{Salariseddigh_GC_IEEE} and entails exploiting of new analysis and inequalities. Here, the error analysis in the achievability proof requires dealing with several combinatorial arguments and using of bounds on the tail of the cumulative distribution function of the Binomial distribution.
\end{itemize}
\subsection{Organization}
The remainder of this paper is structured as follows. In Section~\ref{Sec.SysModel}, system model is explained and the required preliminaries regarding DKI codes are established. Section~\ref{Sec.Res} provides the main contributions and results on the message DKI capacity of the BSC. Finally, Section~\ref{Sec.Summary} of the paper concludes with a summary and directions for future research.
\subsection{Notations}
We use the following notations throughout this paper:

Blackboard bold letters $\mathbbmss{K,X,Y,Z}\ldots$ are used for alphabet sets. Lower case letters $x,y,z,\ldots$ stand for constants and values (realization) of random variables, and upper case letters $X,Y,Z,\ldots$ stand for random variables. Lower case bold symbol $\fx$ and $\fy$ stand for row vectors of size $n$, that is, $\fx = (x_1, \dots, x_n)$ and $\fy = (y_1, \dots, y_n)$. The distribution of a random variable $X$ is specified by a probability mass function (pmf) $p_X(x)$ over a finite set $\X$. The cumulative distribution function (CDF) of a Binomial random variable is indicated by $B_X(x) \triangleq \Pr(X \leq x)$. All logarithms and information quantities are for base $2$. The set of consecutive natural numbers from $1$ through $M$ is denoted by $[\![M]\!]$. The Hamming metric (distance) between two sequences $\mathbf{c}_1$ and $\mathbf{c}_2$ is defined as the number of positions for which the corresponding symbols are not identical, i.e.,
\begin{align}
    d_{\rm H}(\mathbf{x}_1,\mathbf{x}_2) \triangleq \sum_{t=1}^n \delta(x_{i_1,t}, x_{i_2,t}) \,,
\end{align}
where $\delta(.,.)$ the \emph{Kronecker delta} and is defined as follows
\begin{align}
    \delta(x_i,x_j) = 
    \begin{cases}
    1 & x_i \neq x_j
    \\
    0 & x_i = x_j
    \end{cases}
\end{align}
The $n$-dimensional Hamming hyper ball of radius $r$ for integers $n,r$ such that $n \geq r \geq 1$, in the binary alphabet, that is centered at $\mathbf{x}_0 = (x_{0,t})\big|_{t=1}^n$ is defined as
\begin{align}
    \B_{\mathbf{x}_0}(n,r) = \left\{x^n \in \X^n \,:\; d_{\rm H}(\fx,\mathbf{x}_0) \leq r \right\} \;.\,
\end{align}
Volume of the Hamming hyper ball $\B_{\mathbf{x}_0}(n,r)$ in the $q$-ary alphabet is defined as the number of points that lies inside the ball and is denoted by $\text{Vol}\left( \B_{\mathbf{x}_0}(n,r) \right)$. The Hamming cube is defined as the set of sequences with length $n$ and is denoted by $\mathbf{H}^n = \{0,1\}^n$. The $n$-dimensional Hamming hyper ball $\B_{\mathbf{x}_0}(n,r)$ for the choice of $\mathbf{x}_0 = \f0$ and $r=nA$, is denoted by $\B_{\f0}(n,nA)$ and is referred to as the $n$-dimensional Hamming hyper ball in $1$-norm with a corner at the origin, i.e., $\mathbf{0} = (0,\ldots,0)$, and radius equal to $nA$. The definition of $\B_{\f0}(n,nA)$ is given as follows
\begin{align}
    \B_{\f0}(n,nA) = \left\{\fx \in \mathbf{H}^n : 0 \leq \sum_{t=1}^n x_t \leq nA, \forall \, t\in[\![n]\!] \right\} \,.
\end{align}
We use symbol $\triangleq$ to specify a definition convention. The set of whole numbers is denoted by $\mathbb{N}_{0} \triangleq \{0,1,2,\ldots\}$. The $q$-ary entropy function $H_q: [0,1] \to \mathbb{R}$ for $q \geq 2$; a positive integer, is defined as $H_q(\varepsilon) \triangleq x\log_q(q - 1) - x\log_q x - (1 - x)\log_q ( 1 - x)$. The binary entropy function as a special case of the $q$-ary entropy function $H_q(.)$ is denoted by $H(.)$ and defined as $H(\varepsilon) \triangleq - \varepsilon\log(\varepsilon)-(1-\varepsilon)\log(1-\varepsilon)$.

\section{System Model and Preliminaries}
In this section, we present the adopted system model and establish some preliminaries regarding TR and DKI coding.
\label{Sec.SysModel}
\subsection{System Model}
\label{Subsec.SysModel}
We address an identification-focused communication setup, for which the objective of the decoder is defined as follows: Determining whether or not a desired message was sent by the transmitter; see Figure~\ref{Fig.BSC_Channel_TR}. To accomplish this purpose, a coded communication between the transmitter and the receiver over $n$ channel uses of a binary symmetric channel is established\footnote{\,\textcolor{gray!70!black}{The proposed capacity formula works regardless of whether or not a specific code is used for communication, although proper explicit constructed codes may be required to approach the capacity limits.}}. Let $X \in \{0,1\}$ and $Y \in \{0,1\}$ indicate random variables (RVs) which model the input and output of the channel. Each binary input symbol is flipped with probability $0 < \varepsilon < \frac{1}{2}$\footnote{\,\textcolor{gray!70!black}{The extreme cases of $\varepsilon=0$ or $\varepsilon=\frac{1}{2}$ result in $\mathbb{C}_{\rm TR}=1$ and $\mathbb{C}_{\rm TR}=0$, respectively, hence these cases are commonly excluded from the analysis.}}. The stochastic flipping of the input symbol is modelled via an additive Binary Bernoulli noise, i.e., $Z \in \{0,1\}$; see Figure~\ref{Fig.BSC_Channel_TR}. Therefore, the input-output relation of channel reads: $Y = X \oplus Z $, where $\oplus$ indicate the modulo two addition. Throughout the paper, the considered binary symmetric channel with crossover probability $0 < \varepsilon < \frac{1}{2}$ is denoted by $\bW_{\varepsilon}$. We consider the BSC channel $\bW_{\varepsilon}$ which arises as a basic channel model in the context of information theory where the noise distribution, i.e., the probability of observing channel output $Y$ at the receiver given that channel input $X$ was sent at the transmitter, is characterized as follows:
\begin{align}
    W(Y|X) = \begin{cases}
        1 - \varepsilon & Y = X
        \\
        \varepsilon & Y \neq X
    \end{cases}  \,.
\end{align}
for all $x,y \in \{0,1\}$ and $0 < \varepsilon < \frac{1}{2}$.

We assume that $\bW_{\varepsilon}$ is memoryless, that is, the different channel uses are independent. Hence, the transition probability law for $n$ channel uses is given by
\begin{align}
    \label{Eq.Binomial_Channel_Law}
    W^n(\fy|\fx) = \prod_{t=1}^n W(y_t|x_t) = \varepsilon^{d_{\rm H}(\fx,\fy)}(1-\varepsilon)^{n-d_{\rm H}(\fx,\fy)} \,,\,
\end{align}
where $\fx = (x_1,\dots,x_{n})$ and $\fy = (y_1,\dots,y_{n})$ denote the transmitted codeword and the received signal, respectively. Observe that $d_{\rm H}(\fx,\fy)$ is a random variable and follows a Binomial distribution; see Remark\ref{Rem.Hamm_Distan_Dist}.

\subsection{Message Transmission Coding For BSC}
The definition of a TR code for the BSC $\bW_{\varepsilon}$ is given below.
\begin{definition}[BSC-TR Code]
\label{Def.BSCTR-Code}
An $(n,\allowbreak M(n,R),\allowbreak e_1)$-BSC-TR code for a BSC $\bW_{\varepsilon}$ for integer $M(n,R)$, where $n$ and $R$ are the codeword length\footnote{\,\textcolor{gray!70!black}{This code definition restricts itself to codewords of the same length for different messages, which sometimes in the literature is called as the \emph{block} codes. Throughout this paper, we always assume that different TR or DKI codes are \emph{block} codes.}} and coding rate, respectively, is defined as a system $(\C,\sT)$, which consists of a codebook $\C = \{ \mathbf{c}_i \}_{i \in [\![M]\!]}$, with $\fc_i = (c_{i,t})|_{t=1}^n \subset \{0,1\}^n$, such that
\begin{align}
    \frac{1}{n} \sum_{t=1}^n c_{i,t} \leq A \,,
\end{align}
$\forall \, i \in [\![M]\!]$, and a collection of decoders $\sT = \{ \mathbbmss{T}_i \}_{i \in [\![M]\!]}$, where $\mathbbmss{T}_i \subset \{0,1\}^n$, such that the decoders are \emph{mutually disjoint}, i.e.,
\begin{align}
    \mathbbmss{T}_i \cap \mathbbmss{T}_j = \varnothing \,,
\end{align}
for every $i,j \in [\![M]\!]$ such that $i \neq j$. Given a message $i \in [\![M]\!]$, the encoder transmits codeword $\mathbf{c}_i$, and the decoder's task is to address a \emph{multiple hypothesis} as follows: Which message $\hat{i} \in [\![M]\!]$ was sent? There exist one type of error that may happen:
\begin{itemize}
    \item [Error Event:] Rejection of the actual message; $i \in [\![M]\!]$.
\end{itemize}
The associated error probability of the BSC-TR code $(\C,\sT)$ reads
\begin{align}
    \label{Eq.TypeI-Error-TR}
    P_{e,1}(i) & = 1 -\sum_{\fy \in \mathbbmss{T}_{i}} W^n \big( \fy \,\big|\, \fc_i \big) \,,
\end{align}
(see Figure~\ref{Fig.BSC_Channel_DKI}) and satisfy the following bounds $P_{e,1}(i) \leq e_1 \,,\, \forall i \in [\![M]\!] \;, \forall e_1 > 0 \,.$
\end{definition}
Next we define the achievable TR rate and the TR capacity.
\begin{definition}[Achievable Rate]
A TR rate $R>0$ is called achievable if $\forall e_1 > 0$ and sufficiently large $n$, there exists an $(n,\allowbreak M(n\allowbreak,R), \allowbreak e_1)$-BSC-TR code. The operational TR capacity of the BSC $\bW_{\varepsilon}$ is defined as the supremum of all achievable rates, and is denoted by $\mathbb{C}_{\rm TR}(\bW_{\varepsilon},M,K)$.
\end{definition}
\subsection{Message Transmission Capacity of BSC}
In this Subsection we introduce the message transmission problem for the BSC which was introduced originally by Shannon \cite{S48} and introduce some fundamental results on the achievability and the converse proofs.
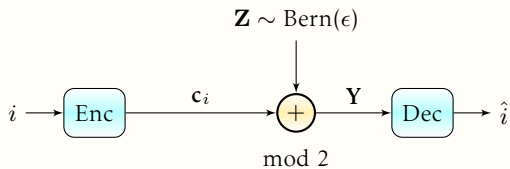
\begin{figure}[!htb]
    \centering
\scalebox{1}{
\tikzstyle{farbverlauf} = [ top color=white, bottom color=white!80!gray]

\tikzstyle{block1} = [draw,top color=white, middle color=cyan!30, rectangle, rounded corners,
minimum height=2em, minimum width=2.5em]

\tikzstyle{block2} = [draw,top color=white, middle color=cyan!30, rectangle, rounded corners,
minimum height=2em, minimum width=2.5em]

\tikzstyle{block_sum} = [draw, top color = white, middle color = orange_6b!30, rectangle, rounded corners, minimum height=1em, minimum width=1.5em]

\tikzstyle{input} = [coordinate]
\tikzstyle{sum} = [draw, circle,inner sep=0pt, minimum size=5mm,  thick]
\tikzstyle{arrow}=[draw,->]

\begin{tikzpicture}[auto, node distance=2cm,>=latex']
\node[] (M) {$\fontfamily{jkpl}\selectfont i$};
\node[block1,right=.5cm of M] (enc) {\text{\fontfamily{jkpl}\selectfont Enc}};
\node[block_sum, sum, right=2cm of enc] (channel) {$+$};

\node[below=1mm of channel] (mod2) {$\text{\fontfamily{jkpl}\selectfont mod 2}$};

\node[block2, right=1cm of channel] (dec) {\text{\fontfamily{jkpl}\selectfont Dec}};
\node[right=.5cm of dec] (Output) {$\fontfamily{jkpl}\selectfont \small \hat{i}$};

\node[above=.7cm of channel] (noise) {$\textbf{Z}\sim \text{\fontfamily{jkpl}\selectfont Bern($\epsilon$)}$};
\draw[->] (M) -- (enc);
\draw[->] (enc) --node[above]{$\fontfamily{jkpl}\selectfont \textbf{c}_i$} (channel);
\draw[->] (noise) -- (channel);
\draw[->] (channel) --node[above]{$\fontfamily{jkpl}\selectfont \textbf{Y}$} (dec);

\draw[->] (dec) -- (Output);

\end{tikzpicture}
}
	\caption{System model of message transmission for the binary symmetric channel. Message $m$ is mapped to the codeword $\fc_i$ where it is flipped with a probability of $\varepsilon$. Decoder employs the output vector $\fy$ to declare a reconstructed version of the original sent message $m$, denoted by $\hat{i}$ in a \emph{reliably} form.}
	\label{Fig.BSC_Channel_TR}
\end{figure}

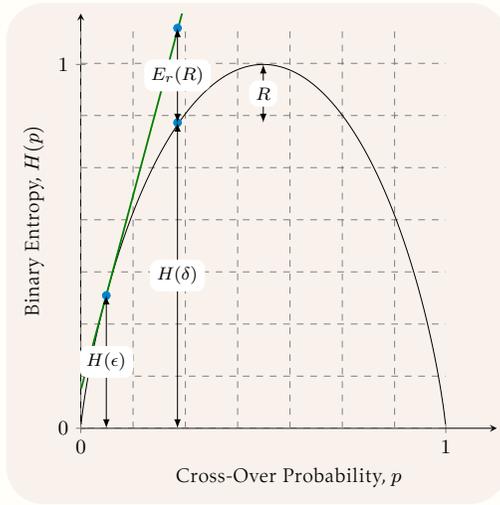
\begin{figure}[t]
    \centering
    \begin{tikzpicture}[scale=.85,background rectangle/.style={fill=brown!10,rounded corners=5mm},show background rectangle]

\begin{axis}[
    axis x line = bottom,
	axis y line = left,
    scale only axis,
    xtick={0,1},
    xticklabels={$0$,$1$},
    ytick={0,1},
    xlabel= \text{\fontfamily{jkpl}\selectfont Cross-Over Probability, $p$} ,
	ylabel= \text{\fontfamily{jkpl}\selectfont Binary Entropy, $H(p)$},
    axis lines=left,
    inner axis line style={latex},
    ymin=0,
    ymax=1.14,
    xmin=0,
    xmax= 1.14,
    width= 65mm,
	height = 65mm,
]
\draw[help lines,dashed,xstep=8.17mm,ystep=8.17mm] (0,0) grid (1,1.1);

\addplot[domain=0.001:0.999,samples = 3000, smooth]({\x},{1*(-1*\x*log2(1*\x)-(1-1*\x)*log2(1-1*\x))});
\addplot[domain=0.001:3,samples = 1000, smooth,mycolor9,thick]({\x},{(-.07*log2(.07)-(1-.07)*log2(1-.07)) + (log2((1 - .07)/.07)*(\x - .07)});


 \filldraw[blau_2b] (.07,.365) circle (1.7pt);
\filldraw[blau_2b] (.265,.84) circle (1.7pt);
\filldraw[blau_2b] (.265,1.1) circle (1.7pt);

\draw[latex-latex] (.07,.365) -- (.07,0) 
    node[below]{$\epsilon$}
    node[midway,fill=white,rounded corners]{\fontfamily{jkpl}\selectfont \small $H(\epsilon)$};

\draw[latex-latex] (.265,.84) -- (.265,0)
    node[below]{$\delta$}
    node[midway,fill=white,rounded corners]{\fontfamily{jkpl}\selectfont \small $H(\delta)$};

\draw[latex-latex] (.265,.84) -- (.265,1.1) node[midway,fill=white,rounded corners]{\fontfamily{jkpl}\selectfont \small $E_r(R)$};

\draw[latex-latex] (.5,.84) -- (.5,1) node[midway,fill=white,rounded corners]{\fontfamily{jkpl}\selectfont \small $R$};

\end{axis}


\end{tikzpicture}
    \caption{Depiction of the error exponent for a BSC. For a given crossover probability $0 < \varepsilon < \frac{1}{2}$, and  The difference between the tangent line to the binary entropy function and the binary entropy function itself is referred to as the error exponent.}
    \label{Fig.Err_Exp}
\end{figure}
\begin{theorem}[see {\cite[Ch.~$5$\,-\,Corrol.~$2$]{RG68}}]
\label{Th.BSC_TR_Capacity}
Assume the binary symmetric channel $\B$ with cros\-s-over probability $\varepsilon$ and message transmission capacity of $\mathbb{C}_{\rm TR}(\W_{\varepsilon}) = 1 - H(\varepsilon)$ and consider an $(n,M(n,R),e_1)$-BSC-TR code for which the codebook size scales exponentially in the codeword length $n$, i.e., $M(n,R)= 2^{nR}$ where $R$ is the message transmission coding rate. Then $0 \leq \forall R < \mathbb{C}_{\rm TR}(\W_{\varepsilon})$ is achievable, namely, there exists a coding and decoding scheme, that is, the existence of a codebook with size $M(n,R) = 2^{nR}$ is guaranteed and the maximum type I error probability converge to zero as $n \to \infty$, i.e.,
\begin{align}
    P_{e,1}(i) \leq 2^{- nE_r(R) + 2 } \,,
\end{align}
for every $i \in [\![M]\!]$, where $E_r(R) > 0$ is a positive, decreasing and convex function\footnote{\,\textcolor{gray!70!black}{The function $E_r(R)$ in the literature is referred to as the \emph{random coding exponent}. Even for the simple BSC, there is no simple way to express the $E_r(R)$ in an analytic functional way for all the values of $0 \leq R< C$ except than in \emph{parametric} form. Cf. \cite{RG68} for further properties of such function. Further discussion are explained in \ref{Eq.Error_Exponent_2}}} of $R$.
\end{theorem}
Here, we restrict ourselves to the following settings: Let the probability assignment on the channel input symbols $0$ and $1$ be the uniform probability mass function. Further, for the sake of accurate analysis of the $E_r(R)$, depending on the range of $\delta$, we divide in two cases as follows:
\begin{numcases}{E_r(R) =}
   \label{Eq.Error_Exponent_1}
     T_{\varepsilon}(\delta) - H(\delta) & $\varepsilon \leq \delta \leq \sqrt{\varepsilon} / ( \sqrt{\varepsilon}+\sqrt{1-\varepsilon} )$
    \\
    1 - R - \log ( \sqrt{\varepsilon} / ( \sqrt{\varepsilon}+\sqrt{1-\varepsilon} ) ) & $\sqrt{\varepsilon} / ( \sqrt{\varepsilon}+\sqrt{1-\varepsilon} ) < \delta < \frac{1}{2}$
    \label{Eq.Error_Exponent_2}
\end{numcases}
where the corresponding values for rate $R$ are given by
\begin{enumerate}
    \item Case $\ref{Eq.Error_Exponent_1} \quad \Rightarrow \quad 1 - H ( \sqrt{\varepsilon} / ( \sqrt{\varepsilon}+\sqrt{1-\varepsilon} ) ) \leq R = 1 - H(\delta) \leq \mathbb{C}_{\rm TR}(\W_{\varepsilon})$
    \item Case $\ref{Eq.Error_Exponent_2} \quad \Rightarrow \quad R < 1 - H ( \sqrt{\varepsilon} / ( \sqrt{\varepsilon}+\sqrt{1-\varepsilon}) )$
\end{enumerate}
\begin{corollary}
\label{Coroll.TR}
For the BSC consider a message set consisting of $2^{m}$ messages\footnote{\,\textcolor{gray!70!black}{Each message is a binary sequence of length $m$, which yields a total of $2^{m}$ message sequences, called the \emph{message set}.}} and let the length of the associated codewords be $n = m(1 + H(\varepsilon) / (1 - H(\varepsilon)) + r)$, where $r > 0$ is an arbitrarily small constant. Then a codebook consisting of codewords with length $n$. Further, the maximum error probability over the entire message set is upper bounded as follows
\begin{align}
    \underset{i \in [\![M]\!]}{\max} \left[ P_{e,1}(i) \right] \leq 2^{-n \alpha(r,\varepsilon) + \log q(r,\varepsilon)} \,.
\end{align}
\end{corollary}
\begin{proof}
Let $\delta > \varepsilon$ be such that
\begin{align}
H(\delta) = \frac{r + (1-r)H(\varepsilon)}{r + 1 - rH(\varepsilon)} \,,
\end{align}
assuming $r>0$ is sufficiently small such that the condition $\delta \leq \sqrt{\varepsilon} / (\sqrt{\varepsilon}+\sqrt{1-\varepsilon})$ as required by Theorem~\ref{Th.BSC_TR_Capacity} is fulfilled. Then the exponent $\alpha(R)$ provided in Theorem~\ref{Th.BSC_TR_Capacity} can be taken to be
\begin{align}
    \label{Eq.Exponent}
    \alpha = T_{\varepsilon}(\delta) - H(\delta) \,.
\end{align}
Observe that the exponent given in \ref{Eq.Exponent} can not be improved by the following theorem.
\end{proof}
\begin{theorem}[see {\cite[Th.~$5.8.5$]{RG68}}]
\label{Th.Wolfowitz_Conv}
Consider the BSC with crossover probability $\varepsilon$; $\W_{\varepsilon}$ with the message transmission capacity of $\mathbb{C}_{\rm TR}(\W_{\varepsilon}) = 1 - H(\varepsilon)$. Further assume an $(n,M(n,\allowbreak R),\allowbreak e_1)$-BSC-TR code where the codebook size scales \textbf{exponentially} in the codeword length $n$, i.e., $2^m = M(m) = M(n,R)= 2^{nR}$ where $R$ is the message transmission coding rate. Now if
\begin{align}
    \label{Ineq.Converse}
    \frac{\log M(n,R)}{n} = R > \mathbb{C}_{\rm TR}(\W_{\varepsilon}) \,,
\end{align}
then the average error probability of BSC-TR code is lower bounded as follows
\begin{align}
    \label{Eq.LB_Ave_Error}
    \bar{P}_{e,1} \geq 1 - \frac{4L}{n(R - \mathbb{C}_{\rm TR}(\W_{\varepsilon}))^2} - 2^{-\frac{n(R - \mathbb{C}_{\rm TR}(\W_{\varepsilon}))}{2}} \,,
\end{align}
where $L > 0$ is a finite positive constant depending on the channel statistics $\varepsilon$ and does not depend on the codeword length\footnote{\,\textcolor{gray!70!black}{The lower bound given in \ref{Eq.LB_Ave_Error} converges to $1$ from left as $n \to \infty$.}}.
\end{theorem}
\begin{corollary}
\label{Coroll.TR_1}
If codeword length $n$ satisfies
\begin{align}
    \label{Ineq.UB_n}
    n < m \left( 1 + \frac{H(\varepsilon)}{1 - H(\varepsilon)} \right) \,,
\end{align}
then there exists \textbf{at least one message} for which the error probability can not be upper bounded by any constant $q<1$.
\end{corollary}
\begin{proof}
Observe that \eqref{Ineq.UB_n} implies a chain of equations as follows
    \begin{align}
        \label{Ineq.n_Chain_Eq}
        n & < m \left( 1 + \frac{H(\varepsilon)}{1 - H(\varepsilon)} \right) \Rightarrow
        \nonumber\\
        & = m \left( \frac{1-H(\varepsilon)+H(\varepsilon)}{1-H(\varepsilon)} \right)
        \nonumber\\
        & = m \left( \frac{1}{1-H(\varepsilon)} \right)
        \nonumber\\
        & = \frac{m}{1-H(\varepsilon)}
        \nonumber\\
        & \stackrel{(a)}{=} \frac{m}{\mathbb{C}_{\rm TR}}
        \nonumber\\
        & = \frac{\log 2^m}{\mathbb{C}_{\rm TR}}
        \nonumber\\
        & \stackrel{(b)}{=} \frac{\log M(m)}{\mathbb{C}_{\rm TR}} \,,
    \end{align}
where $(a)$ employs $\mathbb{C}_{\rm TR}(\W_{\varepsilon}) = 1 - H(\varepsilon)$ and $(b)$ holds by $M(m) = 2^m = M(n,R)$. Thereby exploiting $M(m) = M(n,R)$ into \eqref{Ineq.n_Chain_Eq} yields,
\begin{align}
    \frac{\log M(n,R)}{n} = R > \mathbb{C}_{\rm TR}(\W_{\varepsilon}) \,.
\end{align}
Therefore, by Theorem~\ref{Th.Wolfowitz_Conv} we conclude that the average error probability converges to $1$ which implies that there exist at least one message whose maximum error probability converges to $1$ or can not be upper bounded by any constant $q<1$. This completes the proof of Corollary~\ref{Coroll.TR_1}.
\end{proof}
\begin{corollary}
\label{Coroll.TR_2}
Let $\lambda > 0$ be an arbitrarily finite large constant. Then, $0 < \exists\,\varepsilon < \frac{1}{2}$, such that no code of length $n < \lambda m$ can guarantee any bound $q < 1$ on the maximum error probability of message transmission.
\end{corollary}
\begin{proof}
    Choose an $\varepsilon$ where $0 < \varepsilon < \frac{1}{2}$, so that the following condition is fulfilled
    \begin{align}
        H(\varepsilon) = 1 - \frac{1}{\lambda} \,.
    \end{align}
    Now, apply the Corollary~\ref{Coroll.TR_1}.
\end{proof}
\subsection{DKI Coding For BSC}
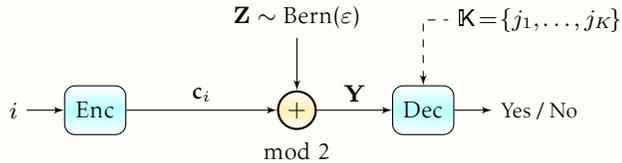
\begin{figure}[!b]
    \centering
    \vspace{.1in}
	\scalebox{1}{
\tikzstyle{farbverlauf} = [ top color=white, bottom color=white!80!gray]

\tikzstyle{block1} = [draw,top color=white, middle color=cyan!30, rectangle, rounded corners,
minimum height=2em, minimum width=2.5em]

\tikzstyle{block2} = [draw,top color=white, middle color=cyan!30, rectangle, rounded corners,
minimum height=2em, minimum width=2.5em]

\tikzstyle{block_sum} = [draw, top color = white, middle color = orange_6b!30, rectangle, rounded corners, minimum height=1em, minimum width=1.5em]

\tikzstyle{input} = [coordinate]
\tikzstyle{sum} = [draw, circle,inner sep=0pt, minimum size=5mm,  thick]
\tikzstyle{arrow}=[draw,->]

\begin{tikzpicture}[auto, node distance=2cm,>=latex']
\node[] (M) {$\fontfamily{jkpl}\selectfont i$};
\node[block1,right=.5cm of M] (enc) {\text{\fontfamily{jkpl}\selectfont Enc}};
\node[block_sum,sum, right=2cm of enc] (channel) {$+$};

\node[below=.5mm of channel] (mod2) {$\text{\fontfamily{jkpl}\selectfont mod 2}$};

\node[block2, right=1cm of channel] (dec) {\text{\fontfamily{jkpl}\selectfont Dec}};
\node[right=.5cm of dec] (Output) {$\text{\fontfamily{jkpl}\selectfont\small Yes\,/\,No}$};

\node[above=.7cm of channel] (noise) {$\textbf{Z}\sim \text{\fontfamily{jkpl}\selectfont Bern($\varepsilon$)}$};

\node[above=.7cm of Output] (Target2) {\fontfamily{jkpl}\selectfont $\mathbbmss{K} \hspace{-.6mm} = \hspace{-.6mm} \{j_1,\hspace{-.2mm}\ldots\hspace{-.2mm},j_K \hspace{-.5mm}\}$};

\draw[->] (M) -- (enc);
\draw[->] (enc) --node[above]{$\fontfamily{jkpl}\selectfont \textbf{c}_i$} (channel);
\draw[->] (noise) -- (channel);
\draw[->] (channel) --node[above]{$\textbf{Y}$} (dec);
\draw[->] (dec) -- (Output);
\draw[dashed,<-] (dec.north) |- (Target2);

\end{tikzpicture}
}
	\vspace{-1mm}
    \captionsetup{justification=justified}
	\caption{System model for DKI communication setting in a BSC. Employing a deterministic encoder in the transmitter, the message $i$ is mapped onto the codeword $\fc_i = (c_{i,t})|_{t=1}^n$ using a deterministic known function. The decoder at the receiver is provided with an arbitrary target message set $\mathbbmss{K}$, and given the channel output vector $\fY = (Y_t)|_{t=1}^n$, it asks whether $i$ belong to $\mathbbmss{K}$ or not.}
    \label{Fig.BSC_Channel_DKI}
\end{figure}
The definition of a DKI code for the BSC $\bW_{\varepsilon}$ is given below.
\begin{definition}[BSC-DKI Code]
\label{Def.BSC-DKI-Code}
An $(n,\allowbreak M(n,R),\allowbreak K(n,\allowbreak \kappa), \allowbreak e_1, \allowbreak e_2)$-BSC-DKI code for a BSC $\bW_{\varepsilon}$ for integers $M(n,R)$ and $K(n,\kappa)$, where $n$ and $R$ are the codeword length and coding rate, respectively, is defined as a system $(\C,\sT_{\mathbbmss{K}})$, which consists of a codebook $\C = \{ \mathbf{c}_i \}_{i \in [\![M]\!]} \subset \{0,1\}^n$, with $\fc_i = (c_{i,t})|_{t=1}^n \subset \{0,1\}^n$, such that
\begin{align}
    \frac{1}{n}\sum_{t=1}^{n} c_{i,t} \leq P_{\,\text{avg}} \,,\,
\end{align}
$\forall i\in[\![M]\!]$ and a decoder\footnote{\,\textcolor{gray!70!black}{We recall that the decoding sets for the DKI problem may have in general intersection, a behaviour similar to that of the RI problem. However to guarantee a vanishing type II error probability we will observe that size of such intersection becomes negligible asymptotically, i.e., as $n \to \infty$}}
\begin{align}
    \sT_{\mathbbmss{K}} = \bigcup_{j \in \mathbbmss{K}} \mathbbmss{T}_j \,,
\end{align}
where $\mathbbmss{T}_j \subset \{0,1\}^n$ is the decoding set corresponding to the single message $\fc_j$, $\forall \mathbbmss{K} \in \left\{ X \subseteq [\![M]\!] \;; |X| = K \right\}$ where $\mathbbmss{K}$ is an arbitrary subset\footnote{\,\textcolor{gray!70!black}{We recall that $\left\{ X \subseteq [\![M]\!] \;; |X| = K \right\}$ is the system (family) of all subsets of the set $[\![M]\!]$, with size $K$. Observe that in general we have $\left| \left\{ X \subseteq [\![M]\!] \;; |X| = K \right\} \right| = \binom{M}{K}$ and the error requirement as imposed by the DKI code definition applies to \emph{each} possible choice of the set $\mathbbmss{K}$ with $K$ arbitrary messages among all $\binom{M}{K}$ cases.}} with size $K$.
Given a message $i \in [\![M]\!]$, the encoder transmits codeword $\mathbf{c}_i$, and the decoder's task is to address a \emph{binary hypothesis}: Was a target message $j \in \mathbbmss{K}$ sent or not? There exist two types of errors that may happen:
\begin{itemize}
    \item [Type I Error Event:] Rejection of the actual message; $i \in \mathbbmss{K}$
    \item [Type II Error Event:] Acceptance of a wrong message; $i \notin \mathbbmss{K}\,.$
\end{itemize}
The associated error probabilities of the DKI code $(\C,\sT)$ reads
\begin{align}
    \label{Eq.TypeI-Error-DKI}
    P_{e,1}(i) & = \Pr \left( \fY \in \sT_{\mathbbmss{K}}^c \,\big|\, \fx = \fc_i \right)_{i \in \mathbbmss{K}} \hspace{-1mm} = 1 - \hspace{-1mm} \sum_{\fy \in \sT_{\mathbbmss{K}}} W^n \big( \fy \, \big| \, \fc_i \big)_{i \in \mathbbmss{K}} \hspace{1mm} \text{(Miss-Identification)} \,,
    \\
    P_{e,2}(i,\mathbbmss{K}) & = \Pr \left( \fY \in \sT_{\mathbbmss{K}} \,\big|\, \fx = \fc_i \right)_{i \notin \mathbbmss{K}} = \sum_{\fy \in \sT_{\mathbbmss{K}}
    } W^n \big( \fy \, \big| \, \fc_i \big)_{i \notin \mathbbmss{K}} \hspace{1mm} \text{(False Identification)} \,.
    \label{Eq.TypeII-Error-DKI}
\end{align}
(see Figure~\ref{Fig.BSC_Channel_DKI}) and satisfy the following bounds
\begin{align}
    P_{e,1}(i) \leq e_1 \,,\, \forall i \in \mathbbmss{K} \,,
    \\
    P_{e,2}(i,\mathbbmss{K}) \leq e_2 \,,\, \forall i \notin \mathbbmss{K} \,.
\end{align}
where $\mathbbmss{K} \in \left\{ X \subseteq [\![M]\!] \;; |X| = K \right\}$ is an arbitrary $K$-size subset of $[\![M]\!]$ and $\forall e_1, e_2>0$.
\end{definition}
\begin{figure}[H]
\centering
\scalebox{1}{
\begin{tikzpicture}[scale=.52
][thick]

\draw[dashed] (0.07,.5) circle (3.5cm);

\draw (0,3.6) circle (.1cm);
\draw (3,1.5) circle (.1cm);
\draw (2,-.1) circle (.1cm);
\draw (-2,2) circle (.1cm);
\draw (-.1,.6) circle (.1cm);
\draw (1.4,-1.4) circle (.1cm);
\draw (-1.8,-.8) circle (.1cm);

\node at (0,3.05) {$\fontfamily{jkpl}\selectfont \fc_2$};
\draw[dashed] (0,3.6) circle (0.2cm);
\draw [fill=cyan!40, fill opacity=0.7] (0,3.6) circle (.1cm);

\node at (3,1) {$\fontfamily{jkpl}\selectfont \fc_3$};

\node at (2,-.5) {$\fontfamily{jkpl}\selectfont \fc_4$};

\node at (-2,1.6) {$\fontfamily{jkpl}\selectfont \fc_1$};
\node at (-.1,.2) {$\fontfamily{jkpl}\selectfont \fc_5$};

\node at (1.4,-1.9) {$\fontfamily{jkpl}\selectfont \fc_6$};
\draw[dashed] (1.4,-1.4) circle (.2cm);
\draw [fill=cyan!40, fill opacity=0.7] (1.4,-1.4) circle (.1cm);

\node at (-1.8,-1.2) {$\fontfamily{jkpl}\selectfont \fc_7$};

\draw[->] (0,4.4) -- ++(0,1)  node [fill=white,inner sep=3pt](a){$\text{\fontfamily{jkpl}\selectfont Input Space}$};

\draw[->] (-2.3+15.07,4.4) -- ++(0,1)  node [fill=white,inner sep=3pt](a){$\text{\fontfamily{jkpl}\selectfont Output Space}$};

\draw[dashed] (-0.3+13.07,.4) circle (3.5cm);

\draw (-4.3+15.50,1.5) circle (1.2cm);    
\draw [fill=gray!20, fill opacity=0.7] (-4.3+15.50,1.5) circle (1.2cm);
\node at (-4.3+15.30,1.5) {$\mathbbmss{T}_{1}$};

\draw (-2.3+15.07,-1.4) circle (1.2cm);    
\draw [fill=gray!20, fill opacity=0.7] (-2.3+15.07,-1.4) circle (1.2cm);
\node at (-2.3+15.07,-1.6) {$\mathbbmss{T}_{6}$};

\draw (-3.3+14.50,-.4) circle (1.2cm);    
\draw [fill=gray!20, fill opacity=0.7] (-3.3+14.50,-.4) circle (1.2cm);
\node at (-3.3+14.30,-.6) {$\mathbbmss{T}_{7}$};

\draw (-0.3+14.50,1.4) circle (1.2cm);     
\draw [fill=gray!20, fill opacity=0.7] (-0.3+14.50,1.4) circle (1.2cm);
\node at (-0.3+14.70,1.4) {$\mathbbmss{T}_{3}$};

\draw (-2.3+15.07,2.4) circle (1.2cm);     
\draw [fill=cyan!20, fill opacity=0.7] (-2.3+15.07,2.4) circle (1.2cm);
\node at (-2.3+15.07,2.6) {$\mathbbmss{T}_{2}$};

\draw (-2.4+15.07,.5) circle (1.2cm);    
\draw [fill=cyan!20, fill opacity=0.7] (-2.4+15.07,.5) circle (1.2cm);
\node at (-2.4+15.07,.5) {$\mathbbmss{T}_{4}$};

\draw (-0.3+14.50,-.5) circle (1.2cm); 
\draw [fill=cyan!20, fill opacity=0.7] (-0.3+14.50,-.5) circle (1.2cm);
\node at (-0.1+14.50,-.5) {$\mathbbmss{T}_{5}$};

\path (0.2,3.7) edge [-> , thick, mycolor9, bend left] node [sloped,midway,above,font=\small] {\text{\fontfamily{jkpl}\selectfont Correct Identification}}(13,3.1);
\path (0.2,3.6) edge [-> , thick, rot_8b, bend right] node [sloped,midway,below,font=\small] {\text{\fontfamily{jkpl}\selectfont Type I Error}}(11.3,2.2);
\path (1.6,-1.4) edge [-> , thick, rot_9b, bend right] node [sloped,midway,below,font=\small] {\text{\fontfamily{jkpl}\selectfont Type II Error}}(13.6,-1.1);

\end{tikzpicture}
}
\captionsetup{justification=justified}
\caption{Depiction of a deterministic 3-identification setting with target message set $\mathbbmss{K} = \{2,3,5\}$. In the correct identification event, channel output is detected in the union of the individual decoder $\mathbbmss{T}_j$ where $j$ belongs to the target message set. Type I error event occurs if the channel output is observed in the complement of the union of individual decoders for which the index of the codeword belongs to. The case where the index of codeword does not coincide to any of the individual decoders for which the channel output belongs to the their union, is referred to as the type II error.}
\label{Fig.DI-Code}
\end{figure}
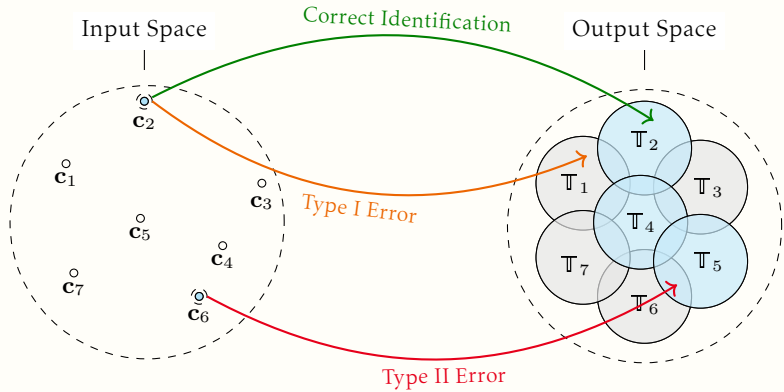
Next we define the achievable DKI rate and the DKI capacity.
\begin{definition}[Achievable Rate]
A DKI rate $R>0$ is called achievable if $\forall e_1, \allowbreak e_2>0$ and sufficiently large $n$, there exists an $(n,\allowbreak M(n\allowbreak,R),\allowbreak K(n,\allowbreak \kappa), \allowbreak e_1, \allowbreak e_2)$ DKI code. The operational DKI capacity of the BSC $\bW_{\varepsilon}$ is defined as the supremum of all achievable rates, and is denoted by $\mathbb{C}_{\rm DKI}(\bW_{\varepsilon},M,K)$.
\end{definition}
\begin{remark}
\label{Rem.Hamm_Distan_Dist}
Assuming that codeword $\fc_i$ is sent and channel output $\fy$ is observed at receiver, the number of crossovers (flips) that occur in the channel is $d_{\rm H}(\fy,\fc_i)$. Therefore, the probability that $k$ crossovers among the $n$ channel uses occur, follows a Binomial distribution with parameters $n$ and $\varepsilon$, as follows
\begin{align}
    \label{Eq.Hamm_Distan_Dist}
    \Pr \left( d_{\rm H}(\fY,\fc_i) = k \right) = \binom{n}{k} \varepsilon^k (1-\varepsilon)^{n-k} \,,
\end{align}
\end{remark}
\section{DKI Capacity of The BSC}
\label{Sec.Res}
In this section, we present our main results, i.e., achievability and converse proofs for the BSC. Subsequently, we provide the detailed proofs.
\subsection{Main Results}
The DKI capacity theorem for the BSC channel $\bW_{\varepsilon}$ is stated below.
\begin{theorem}
\label{Th.DKI-Capacity}
Let $\W_{\varepsilon}$ indicate a BSC with cross-over probability $0<\varepsilon<\frac{1}{2}$ and let $\beta > 0$ be an arbitrary small positive. Further let $H(p)$ indicate the binary entropy function and $T_{\varepsilon}(p) = H(\varepsilon) + (p - \varepsilon) \diff{H(p)}{p}|_{p=\varepsilon}$ specify the tangent to $H(p)$ at point $\varepsilon \,.$ Then the deterministic K-Identification capacity of $\W_{\varepsilon}$ subject to the Hamming weight constraint of the form $n^{-1}\sum_{t=1}^n x_t \leq A$ with \textbf{exponentially} large target message set, i.e., $K(n,\kappa) = 2^{\kappa n}$ where the target identificaiton rate $\kappa$ satisfy
\begin{align}
    \label{Ineq.kappa_theorem}
    0 \leq \kappa < T_\varepsilon\left( \left( 1 - \beta \right) \, \varepsilon + \beta / 2 \right) - H\left( \left( 1 - \beta \right) \, \varepsilon + \beta / 2 \right)
    \,,
\end{align}
and in the \textbf{exponential} codebook size, i.e., $M(n,R) = 2^{nR}$, is given by
\begin{align}
    \mathbb{C}_{\rm DKI} \left( \W_{\varepsilon},K \right) & = \begin{cases}
    H(A) & \text{if $A < \frac{1}{2}$}\\
    1 & \text{if $A \geq \frac{1}{2}$} \,. 
    \end{cases}
\end{align}
\end{theorem}
\vspace{-1mm}
\begin{proof}
The proof of Theorem~\ref{Th.DKI-Capacity} consists of two parts, namely the achievability and the converse proofs, which are provided in Sections~\ref{Sec.Achiev} and \ref{Sec.Conv}, respectively.
\end{proof}
\subsection{Lower Bound (Achievability Proof)}
\label{Sec.Achiev}
The achievability proof consists of the following two main steps.
\begin{itemize}
    \item \textbf{\textcolor{mycolor12}{Step~1:}} First, we propose a greedy-wise codebook construction and derive an analytical lower bound on the corresponding codebook size using similar argument as provided in the Gilbert-Varshamov (GV) bound\footnote{\,\textcolor{gray!70!black}{The early introduction of such bound in the literature is accomplished by Gilbert in \cite{Gilbert52}.}} for packing of non-overlapping balls embedded in the input space.
    \item \textbf{\textcolor{mycolor12}{Step~2:}} Then, to prove that this codebook leads to an achievable rate, we propose a decoder and show that the corresponding type I and type II error rates vanished as $n \to \infty$.
\end{itemize}
\subsubsection*{Codebook Construction}
Let $A = P_{\,\text{ave}}$. In the following, we confine ourselves to codewords that meet the condition $n^{-1} \sum_{t=1}^n c_{i,t} \leq A$, $\forall \, i \in [\![M]\!]$.
\begin{itemize}
    \item \textbf{\textcolor{mycolor12}{Case 1 - With Hamming Weight Constraint:}}
    $A \leq 1$, then the condition $n^{-1} \sum_{t=1}^n \allowbreak c_{i,t} \allowbreak \leq 1 \,, i \in [\![M]\!]$ is non-trivial in the sense that it induces a strict subset of the entire input space $\mathbf{H}^n$. We denote such subset by $\B_{\f0}(n,nA)$ and is equivalent to $\norm{\fc_i}_1 \leq A$.
    \item \textbf{\textcolor{mycolor12}{Case 2 - Without Hamming Weight Constraint:}}
    $A \geq 1$, then each codeword belonging to the $n$-dimensional Hamming cube $\mathbf{H}^n$ fulfilled the Hamming weight constraint, since $\frac{1}{n} \sum_{t=1}^n c_{i,t} \leq 1 \leq A \,, i \in [\![M]\!]$. Therefore, we address the entire input space $\mathbf{H}^n = \{0,1\}^n$ as the possible set of codewords and attempt to exhaust it in a brute-force manner in the asymptotic, i.e., as $n \to \infty$.
\end{itemize}
$\rhd$ \textbf{Analysis For Case 1:}
Observe that within this case, we again divide into two cases:
\begin{enumerate}
    \item $0 < A < \frac{1}{2}$
    \item $A \geq \frac{1}{2}$
\end{enumerate}
The argument for the need of such division is that the binary entropy function is monotonic increasing only for $0 \leq A \leq \frac{1}{2}$ and for $A \geq \frac{1}{2}$ is decreasing. That is, in the latter case, we can introduce an alternative Bernoulli process which result in a larger volume space, and at the same time, it guarantees the Hamming weight constraint.

For the sub-case 1, i.e., where $0 < A < \frac{1}{2}$, we restrict our considerations to an $n$-dimensional Hamming hyper ball with edge length $A$. We use a packing arrangement of overlapping hyper balls of radius $r_0 = \floor{n\beta}$ in an $n$-dimensional Hamming hyper ball $\B_{\f0}(n,nA)$, where

\begin{lemma}
\label{Lem.Exhaustion_Case_1}
Let $R < H(A)$ and let $\beta > 0$ be an arbitrary small constant.
Then for sufficiently large codeword length $n$, there exist a codebook $\C = \{ \mathbf{c}_i \}_{i \in [\![M]\!]} \subset \{0,1\}^n$, with $\fc_i = (c_{i,t})|_{t=1}^n \subset \{0,1\}^n$, which consists of $M$ sequences in the $n$-dimensional Hamming hyper ball $\B_{\f0}(n,nA)$, such that the following holds:
\begin{itemize}
    \item[\text{Hamming Distance Property}:] $\quad d_{\rm H}(\fc_{i},\fc_{j}) \geq \floor{n\beta} + 1 \qquad \forall i,j \in [\![M]\!]\;$ where $\;i \neq j \,.$
    \item[Codebook Size:] The codebook size is at least $\; 2^{nR-1}$, that is, $M \geq 2^{n(R-\frac{1}{n})}\,.$
\end{itemize}
\end{lemma}
\begin{proof}
Recall that the minimum Hamming distance of a code $\C$ is given by
\begin{align}
    d_{\,\text{min}} \triangleq \underset{(i,j) \in [\![M]\!] \times [\![M]\!] }{\min} \, d_{\rm H}(\fc_i,\fc_j) \;.
\end{align}
We begin to obtain some codeword that fulfill the Hamming weight constraint, namely,
\begin{align}
    \label{Ineq.Hamming_Constraint}
    \frac{1}{n} \sum_{t=1}^n c_t \leq A \,.
\end{align}
First, we generate a codeword $\fC \overset{i.i.d}{\sim} \text{Bernoulli}(A)$\footnote{\,\textcolor{gray!70!black}{Such a random generation should not be confused with a similar procedure as is accomplished in the encoding stage of the RI problem. While therein, each message is mapped to a codeword through a random distribution, here for the DI problem, we first solely restrict ourselves to generation of codewords through the Bernoulli distribution to guarantee the Hamming weight constraint, and employ them in the next procedure called the greedy construction up to an exhaustion. Then, after the exhaustion, we establish a deterministic mapping between the message set and the codebook, that is, each message is associated to a codeword. Further, in the RI problem, it is in general possible that two different message are mapped to a common codeword, however, considering the DI problem in here, there exist a one-to-one mapping between the set of messages and the set of codewords.}}. Since $\mathbb{E}\left[ C_t \right] = A$, by the \emph{weak law of large numbers}, we obtain
\begin{align}
    \lim_{n\to\infty} \Pr \left( \left| \frac{1}{n} \sum_{t=1}^n C_t - A \right| \leq \tau \right) = 1 \,,
\end{align}
where $\tau>0$ is an arbitrary small positive. Therefore, for sufficiently large codeword length $n$, the event $\big| n^{-1} \sum_{t=1}^n C_t - A \big| \leq \tau$ occurs with probability $1$, which implies
\begin{align}
    \label{Ineq.Event_Hamming_Constraint}
    \frac{1}{n} \sum_{t=1}^n C_t \leq A + \tau \,.
\end{align}
Now, observe that since \eqref{Ineq.Event_Hamming_Constraint} holds for arbitrary values of $\tau$, it implies that the following condition for sufficiently large $n$, is fulfilled
\begin{align}
    \frac{1}{n} \sum_{t=1}^n C_t \leq A \,,
\end{align}
which is the Hamming weight constraint as required.
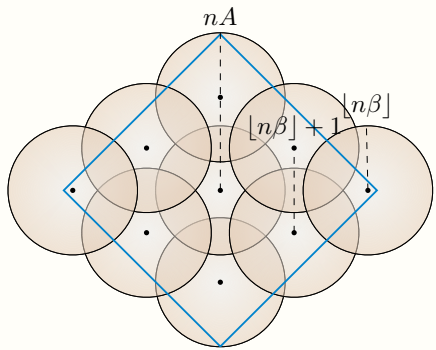
\begin{figure}[t]
    \centering
	\scalebox{.87}{
\begin{tikzpicture}[scale=.65,rotate=45][thick]

\draw (-1.53,-1.53) circle (1.52cm);
\draw [inner color =gray!15,outer color = brown!40, fill opacity=0.2] (-1.53,-1.53) circle (1.52cm);
\node [fill=black, shape=circle, inner sep=.8pt] ($.$) at (-1.53,-1.53) {};
\draw (0,0) circle (1.52cm);
\draw [inner color =gray!15,outer color = brown!40, fill opacity=0.2] (0,0) circle (1.52cm);
\draw (1.55,1.55) circle (1.52cm);
\draw [inner color =gray!15,outer color = brown!40, fill opacity=0.2] (1.55,1.55) circle (1.52cm);
\node [fill=black, shape=circle, inner sep=.8pt] ($.$) at (1.55,1.55) {};
\draw (+.52,-1.93) circle (1.52cm);
\draw [inner color =gray!15,outer color = brown!40, fill opacity=0.2] (+.52,-1.93) circle (1.52cm);
\node [fill=black, shape=circle, inner sep=.8pt] ($.$) at (+.52,-1.93) {};
\draw (1.93,-.52) circle (1.52cm);
\draw [inner color =gray!15,outer color = brown!40, fill opacity=0.2] (1.93,-.52) circle (1.52cm);
\node [fill=black, shape=circle, inner sep=.8pt] ($.$) at (1.93,-.52) {};
\draw (-1.93,+.52) circle (1.52cm);
\draw [inner color =gray!15,outer color = brown!40, fill opacity=0.2] (-1.93,.52) circle (1.52cm);
\node [fill=black, shape=circle, inner sep=.8pt] ($.$) at (-1.93,.52) {};
\draw (-.52,1.93) circle (1.52cm);
\draw [inner color =gray!15,outer color = brown!40, fill opacity=0.2] (-.52,1.93) circle (1.52cm);
\node [fill=black, shape=circle, inner sep=.8pt] ($.$) at (-.52,1.93) {};

\draw (2.45,-2.45) circle (1.52cm);
\draw [inner color =gray!15,outer color = brown!40, fill opacity=0.2] (2.45,-2.45) circle (1.52cm);
\node [fill=black, shape=circle, inner sep=.8pt] ($.$) at (2.45,-2.45) {};
%
%
\draw (-2.45,2.45) circle (1.52cm);
\draw [inner color =gray!15,outer color = brown!40, fill opacity=0.2] (-2.45,2.45) circle (1.52cm);
\node [fill=black, shape=circle, inner sep=.8pt] ($.$) at (-2.45,2.45) {};
%

\foreach \s in {2.6}
{
\draw [thick,blau_2b] (-\s,-\s) -- (\s,-\s) -- (\s,\s) -- (-\s,\s) -- (-\s,-\s);
}

\node [fill=black, shape=circle, inner sep=.8pt] ($.$) at (0,0) {};



\draw [dashed] (+.52,-1.93) -- (1.93,-.52) node [above,font=\large] {$\floor{n\beta}+1$};
\draw [dashed] (0,0) -- (2.63,+2.63) node [above,font=\large] {$nA$};
\draw [dashed] (2.45,-2.45) -- (3.5,-1.35) node [above,font=\large] {$\floor{n\beta}$};
\end{tikzpicture}}
	\vspace{3mm}
	\caption{Illustration of an exhausted greedy-wise ball packing inside a hyper cube in $1$-norm, where union of the small balls of radius $r_0 = \floor{n\beta}$ cover a larger cube. As the codewords are assigned to the center of each ball lying inside the hyper cube according to the greedy construction, the $1$-norm of a codeword is bounded by $nA$ as required.}
	\vspace{-5mm}
	\label{Fig.Density}
\end{figure}

Next, we begin with the greedy procedure as follows: Let denote the first codeword determined by the Bernoulli distribution by $\fc_1$ and assign it to message with index $1$. Then, we remove all the sequences that have a Hamming distance of less or equal than $\floor{n\beta}$ from $\fc_1$. That is, we delete all the codewords that lies inside the Hamming ball with center $\fc_1$ and radius $r = \floor{n\beta}$. Then, we generate a second codeword by the Bernoulli distribution and repeat this procedure until all the sequences belonging to the legit subspace, i.e., the Hamming hyper ball in $1$-norm; $\B_{\f0}(n,nA)$, are exhausted. Therefore, such a construction fulfill the property provided in Lemma~\ref{Lem.Exhaustion_Case_1} regarding the minimum Hamming distance of the code, i.e.,
\begin{align}
    \label{Ineq.d_H_LB_1}
    d_{\rm H}(\fc_{i},\fc_{j}) \geq \floor{n\beta} + 1 \;.
\end{align}
In general, the volume of a Hamming ball of radius $r$, assuming that the alphabet size is $q$, is the number of codewords that it encompass and is given by \cite[see Ch.~1]{Richardson08}
\begin{align}
    \label{Eq.Hamming_Ball_Volume}
    \text{Vol}\left(\B_{\fx}(n,r)\right) = \sum_{i=0}^{r} \binom{n}{i} (q-1)^i \,.
\end{align}
Let $\mathscr{B}$ denote the obtained ball packing after the exhaustion of the entire Hamming hyper ball $\B_{\f0}$, i.e., an arrangement of $M$ overlapping small hyper balls $\B_{\fc_i}(n,r_0)$, with radius $r_0 = \floor{n\beta}$ where $i \in [\![M]\!]$, that cover the entire Hamming hyper ball in $1$-norm; $\B_{\f0}(n,nA)$, where their centers are coordinated inside the $\B_{\f0}(n,nA)$, and the distance between the closest centers is $\floor{n\beta} + 1$; see Figure~\ref{Fig.Density}. As opposed to the standard ball packing observed in coding techniques \cite{CHSN13}, the balls here are neither necessarily entirely contained within the Hamming hyper ball, nor disjoint. That is, we only require that the centers of the balls are inside $\B_{\f0}(n,nA)$ and have a non-empty intersection with $\B_{\f0}(n,nA)$, which is rather a \emph{ball covering problem}.

Th ball packing $\mathscr{B}$ is called \emph{exhausted} if no point within the input space; $\B_{\f0}(n,nA)$, is remained as an \emph{isolated point}, that is, with the property that it does not belong to \emph{at least} one of the small Hamming hyper balls.

In particular, we use a covering argument that has a similar flavor as that observed in the GV bound \cite[Th.~5.1.7]{Van98}. Specifically, consider an exhausted packing arrangement of
\begin{align}
    \label{Eq.Union_balls}
    \bigcup_{i=1}^{M(n,R)} \B_{\fc_i}(n,\floor{n\beta}) \,,
\end{align}
balls with radius $r_0 = \floor{n\beta}$ embedded within the space $\B_{\f0}(n,nA)$. According to the greedy construction, the center $\fc_i$ of each small Hamming hyper ball, corresponds to a codeword. Since the volume of each hyper ball is equal to $\text{Vol}(\B_{\fc_1}(n,r_0))$, the centers of all balls lie inside the space $\B_{\f0}(n,nA)$, and the Hamming hyper balls \emph{overlap} with each other, the total number of balls is bounded from below by
\begin{align}
    \label{Ineq.M_Achiev}
    M & \geq \frac{\text{Vol}\left(\bigcup_{i=1}^{M}\B_{\fc_i}(n,r_0)\right)}{\text{Vol}(\B_{\fc_1}(n,r_0))}
    \nonumber\\&
    \stackrel{(a)}{\geq} \frac{\text{Vol} \left( \B_{\f0}(n,nA) \right)}{\text{Vol}(\B_{\fc_1}(n,r_0))}
    \nonumber\\&
    \stackrel{(b)}{\geq} \frac{\sum_{j=0}^{\floor{nA}} \binom{n}{j}}{\text{Vol}(\B_{\fc_1}(n,r_0))} \,,
\end{align}
where $(a)$ holds since the Hamming hyper balls may have in general \emph{intersection} and $(b)$ follows by \eqref{Eq.Hamming_Ball_Volume} with setting $q=2$ and since $\floor{nA} \leq nA$. Now, the bound in \eqref{Ineq.M_Achiev} can be further simplified as follows 
\begin{align}
    \label{Ineq.Log_M_Achiev_1}
    \log M & \geq \log \left( \frac{\sum_{j=0}^{\floor{nA}} \binom{n}{j}}{\text{Vol}(\B_{\fc_1}(n,r_0))} \right)
    \nonumber\\
    & \stackrel{(a)}{\geq} nH(A) + o\left( \log n \right) - nH(\beta)
    \,.
\end{align}
where $(a)$ exploits \eqref{Ineq.Bound_Volume_Hamming_ball_LB} for setting radius $r = \floor{n\varepsilon} = \floor{nA}$ and $q=2\,,$ and \eqref{Ineq.Bound_Volume_Hamming_ball_UB} with $r_0 = \floor{n\varepsilon} = \floor{n\beta} \,.$ Now, we obtain
\begin{align}
    \label{Ineq.Log_M_Achiev_2}
    \log M & \geq nH(A) + o\left( \log n \right) - nH \left( \beta \right)
    \,,
\end{align}
where the dominant term has an order of $n$. Therefore, in order to obtain finite value for the lower bound on the DKI rate, $R$, \eqref{Ineq.Log_M_Achiev_2} induces the scaling law of codebook size, $M$, to be $2^{nR}$. Hence, we obtain
\begin{align}
    R & \geq \frac{1}{n} \left[ nH(A) + o\left( \log n \right) - nH \left( \beta \right) \right] \,,
    \nonumber\\
    & = H(A) + \frac{o\left( \log n \right)}{n} - H \left( \beta \right) \,,
\end{align}
which tends to $H(A)$ as $n \to \infty$ and $\beta \to 0\,.$

Now, we proceed to the sub-case 2, i.e., where $A \geq \frac{1}{2}\,.$ In this case, instead of sticking to generation of codewords $\sim \text{Bernoulli}(A)$, we generate the codewords according to Bernoulli process with success probability of $\frac{1}{2}$, that is, $\fC \overset{i.i.d}{\sim} \text{Bernoulli}(\frac{1}{2})$. Observe that now the required Hamming weight constraint given in \eqref{Ineq.Hamming_Constraint} is met, since for $\mathbb{E}\left[ C_t \right] = \frac{1}{2}$ we have
\begin{align}
    \label{Ineq.Hamming_Weight_Altenative_Experiment}
    \frac{1}{n} \sum_{t=1}^n c_t \leq \frac{1}{2} \leq A \,.
\end{align}
Therefore, following similar line of arguments as provided for the sub-case 1, we obtain the following lower bound on the DKI rate, $R$,
\begin{align}
    R & \geq \frac{1}{n} \left[ nH\left(\frac{1}{2}\right) + o\left( \log n \right) - nH \left( \beta \right) \right] \,,
    \nonumber\\
    & = H\left(\frac{1}{2}\right) + \frac{o\left( \log n \right)}{n} - H \left( \beta \right) \,,
\end{align}
which tends to $H\left(\frac{1}{2}\right) = 1$ as $n \to \infty$ and $\beta \to 0\,.$
\end{proof}
$\rhd$ \textbf{Analysis For Case 2:}

\begin{lemma}[see {\cite[Claim~$1$]{J85}}]
\label{Lem.Exhaustion}
The entire Hamming cube $\mathbf{H}^n$ can be exhausted asymptotically as the codebook, that is, all the message sequences, i.e., the indices between $1$ and $2^m$ can be coded with binary sequences of length $n$, subject to the Hamming distance property, i.e.,
\begin{align}
    \label{Ineq.d_H_LB_2}
    d_{\rm H}(\fc_{i},\fc_{j}) \geq \floor{n\beta} + 1 \;.
\end{align}
for every $i,j \in [\![M]\!]$, where $i \neq j$ and with $\beta > 0$ being an arbitrary small positive.
\end{lemma}
\begin{proof}
Recall that the minimum Hamming distance of a code $\C$ is given by
\begin{align}
    d_{\,\text{min}} \triangleq \underset{(i,j) \in [\![M]\!] \times [\![M]\!] }{\min} \, d_{\rm H}(\fc_i,\fc_j) \;.
\end{align}
Next, we begin with the greedy procedure as follows: Let denote the first codeword determined by the Bernoulli distribution by $\fc_1$ and assign it to message with index $1$. Then, we remove all the sequences that have a Hamming distance of less or equal than $\floor{n\beta}$ from $\fc_1$. That is, we delete all the codewords that lies inside the Hamming ball with center $\fc_1$ and radius $r = \floor{n\beta}$. Then, we generate a second codeword by the Bernoulli distribution and repeat this procedure until all the sequences are exhausted.

Let $\mathscr{B}$ denote the obtained ball packing after the exhaustion of the entire input space $\mathbf{H}^n = \{0,1\}^n$, i.e., an arrangement of $M$ overlapping small hyper balls $\B_{\fc_i}(n,r_0)$, with radius $r_0 = \floor{n\beta}$ where $i\in[\![M]\!]$, that cover $n$-dimensional Hamming cube $\mathbf{H}^n = \{0,1\}^n$, where their centers are coordinated inside $\mathbf{H}^n$, and the distance between the closest centers is $\floor{n\beta} + 1$. As opposed to the standard ball packing observed in coding techniques \cite{CHSN13}, the balls here are neither necessarily entirely contained within the Hamming hyper ball, nor disjoint. That is, we only require that the centers of the balls are inside $\mathbf{H}^n$ and have a non-empty intersection with $\mathbf{H}^n$, which is rather a \emph{ball covering problem}. The ball packing $\mathscr{B}$ is called \emph{exhausted} if no point within the input space; $\mathbf{H}^n$, is remained as an \emph{isolated point}, that is, with the property that it does not belong to \emph{at least} one of the small Hamming hyper balls.

In particular, we use a covering argument that has a similar flavor as that observed in the GV bound \cite[Th.~5.1.7]{Van98}. Specifically, consider an exhausted packing arrangement of
\begin{align}
    \bigcup_{i=1}^{M(n,R)} \B_{\fc_i}(n,\floor{n\beta}) \,,
\end{align}
balls with radius $r_0 = \floor{n\beta}$ embedded within the space $\mathbf{H}^n$. According to the greedy construction, the center $\fc_i$ of each small Hamming hyper ball, corresponds to a codeword. Since the volume of each hyper ball is equal to $\text{Vol}(\B_{\fc_1}(n,r_0))$, the centers of all balls lie inside the space $\mathbf{H}^n$, and the Hamming hyper balls \emph{overlap} with each other, the total number of balls is bounded from below by
\begin{align}
    \label{Ineq.M_Achiev_2}
    M & \geq \frac{\text{Vol}\left(\bigcup_{i=1}^{M}\B_{\fc_i}(n,r_0)\right)}{\text{Vol}(\B_{\fc_1}(n,r_0))}
    \nonumber\\&
    \stackrel{(a)}{\geq} \frac{\text{Vol} \left( \mathbf{H}^n \right)}{\text{Vol}(\B_{\fc_1}(n,r_0))}
    \nonumber\\&
    \stackrel{(b)}{\geq} \frac{|\X|^n}{\text{Vol}(\B_{\fc_1}(n,r_0))} \,,
\end{align}
where $(a)$ holds since the Hamming hyper balls may have in general \emph{intersection} and $(b)$ follows since $\text{Vol} \left( \mathbf{H}^n \right) = \left| \X^n \right| = |\X|^n$. Now, the bound in \eqref{Ineq.M_Achiev_2} can be further simplified as follows 
\begin{align}
    \label{Ineq.Log_M_Achiev_2_1}
    \log M & \geq \log \left( \frac{|\X|^n}{\text{Vol}(\B_{\fc_1}(n,r_0))} \right)
    \nonumber\\
    & \stackrel{(a)}{\geq} n\log |\X| + o\left( \log n \right) - n H(\beta)
    \nonumber\\
    & \stackrel{(b)}{\geq} n + o\left( \log n \right) - n H(\beta)
    \,.
\end{align}
where $(a)$ exploits \eqref{Ineq.Bound_Volume_Hamming_ball_UB} with $\varepsilon = \beta\,.$ Now for $\beta > 0$ being an arbitrary small positive, we obtain
\begin{align}
    \label{Ineq.Log_M_Achiev_2_2}
    \log M & \geq n + o\left( \log n \right) - n H\left( \beta \right)
    \nonumber\\
    & = n \left( 1 - H \left( \beta \right) \right) + o\left( \log n \right)
    \,,
\end{align}
where the dominant term has an order of $n$. Therefore, in order to obtain finite value for the lower bound on the DKI rate, $R$, \eqref{Ineq.Log_M_Achiev_2} induces the scaling law of codebook size, $M$, to be $2^{nR}$. Hence, we obtain
\begin{align}
    R & \geq \frac{1}{n} \left[ n \left( 1 - H \left( \beta \right) \right) + o\left( \log n \right) \right] \,,
    \nonumber\\
    & = 1 - H \left( \beta \right) + \frac{o\left( \log n \right)}{n} \,,
\end{align}
which tends to $1$ as $n \to \infty$ and $\beta \to 0\,.$
\end{proof}
\subsubsection*{Encoding}
Given a message $i \in [\![M]\!]$, transmit $\fx = \fc_i$.
\subsubsection*{Decoding}
Let define $\delta_{\beta} \neq \frac{1}{2}$ as follows
\begin{align}
    \label{Ineq.delta_1}
    \delta_{\beta} = \left( 1 - \beta/2 \right)\, \varepsilon + \beta/4 \,,
\end{align}
which is referred to as the \emph{decoding threshold} with $\beta > 0$ being an arbitrary small. Observe that given $0 < \varepsilon < 1/2$ and \eqref{Ineq.delta_1}, we obtain the following bounds on the $\delta_{\beta}$
\begin{align}
    \label{Ineq.delta_2}
     \varepsilon < \delta_{\beta} < ( 1 - \beta )\, \varepsilon + \beta / 2 \,.
\end{align}
To identify whether message $j \in [\![M]\!]$ was sent, the decoder checks whether the channel output $\fy$ belongs to the decoding set $\sT_{\mathbbmss{K}} = \bigcup_{j \in \mathbbmss{K}} \mathbbmss{T}_j$, or not, with
\begin{align}
    \label{Def.Dec_DKI}
    \mathbbmss{T}_j = \Big\{ \fy \in \mathbf{H}^n \;;\, T(\fy,c_j) \leq \floor{n\delta_{\beta}} \Big\} \,,
\end{align}
where
\begin{align}
    \label{Eq.Decoding_Metric}
    T(\fy,c_j) = d_{\rm H}(\fy,\fc_j) \triangleq \sum_{t=1}^n \delta_{\beta}(y_t, c_{j,t}) \,,
\end{align}
is referred to as the \emph{decoding metric} evaluated for observation vector $\fy$ and the individual codeword $\fc_j$, with $\delta_{\beta}(.,.)$ being the \emph{Kronecker delta}. That is, given the channel output vector $\fy \in \mathbf{H}^n$, if there exist at least one $j \in \mathbbmss{K}$ such that $d_{\rm H}(\fy,\fc_j) \leq \floor{n\delta_{\beta}}$, then the decoder declares that the message $j$ was sent. And for the other case, i.e., where for each index $j \in \mathbbmss{K}$, the inequality $d_{\rm H}(\fy,\fc_j) > \floor{n\delta_{\beta}}$ holds, then the decoder decides that $j$ was not sent.
\begin{remark}
\textbf{Adopted Decoder}:
For the achievability proof, we adopt a decoder which upon observing an output sequence $\fy$, it declares that the message $j \in \mathbbmss{K}$ was sent if the output vector $\fy$ belongs to the following set
\begin{align}
    \bigcup_{j \in \mathbbmss{K}}  \Big\{ \fy \in \mathbf{H}^n \;;\, d_{\rm H}(\fy,\fc_j) \leq \floor{n\delta_{\beta}} \Big\}
    \;,\,
    \label{Eq.Decoding_Set_0}
\end{align}
where $\fc_j = [c_{j,1},\ldots,c_{j,n}]$ is the codeword associated with message $j$ and $\delta_{\beta}$ is a decoding threshold. We notice that the decoder in \eqref{Eq.Decoding_Set_0} combine the elements of set $\mathbbmss{K}$ through a fundamental union operator. Such a simple operator may feature a penalty with respect to the error exponents for the type I/II error probabilities or the obtained achievable rates. Therefore, we recall that in principle a more optimum decoder for the K-Identification scheme which guarantee vanishing type I/II error probabilities, might demand a more complicated algebraic operators between the realization of members for each specific set $\mathbbmss{K}$ and entails advanced dependencies on the elements of set $\mathbbmss{K}$.
\end{remark}
\subsubsection*{Error Analysis}
Fix $e_1,e_2 > 0$ and let $\zeta_0, \zeta_1 > 0$ be arbitrarily small constants. Further, let introduce the following conventions:
\begin{itemize}
    \item $Y_t(i)$ denote the channel output at time $t$ \emph{conditioned} that the sent codewords was $\fx = \fc_i$, that is, $Y_t(i) = \fc_{i,t} \oplus Z_t$
    \item The output vector is defined as the vector of symbols, i.e., $\fY(i) \triangleq (Y_1(i),\ldots,Y_n(i))$
\end{itemize}
Consider the type I error, i.e., the transmitter sends $\fc_i$, yet $\fy \notin \sT_{\mathbbmss{K}}$ for every $i \in \mathbbmss{K}$. The type I error probability is given by
\begin{align}
    \label{Eq.Type_I_Error}
    P_{e,1}(i) & = \Pr\left( \fY(i) \in \sT_{\mathbbmss{K}}^c \right)
    \nonumber\\
    & = \Pr \left( \fY(i) \in \left( \bigcup_{j \in \mathbbmss{K}} \mathbbmss{T}_j \right)^c \right)
    \nonumber\\
    & \stackrel{(a)}{=} \Pr \left( \fY(i) \in \bigcap_{j \in \mathbbmss{K}} \mathbbmss{T}_j^c \right)
    \nonumber\\
    & \stackrel{(b)}{\leq} \Pr \left( \fY(i) \in \mathbbmss{T}_i^c \right)
    \nonumber\\
    & = \Pr\left( T(\fY(i),c_i) > \floor{n\delta_{\beta}} \right)
    \,,
\end{align}
where $(a)$ holds by \emph{De Morgan}'s law for a finite number of union of set, i.e., $\left( \bigcup_{i \in \mathbbmss{K}} \mathbbmss{T}_i \right)^c = \bigcap_{i \in \mathbbmss{K}} \mathbbmss{T}_i^c$ and $(b)$ follows since $\bigcap_{j \in \mathbbmss{K}} \mathbbmss{T}_j^c \subset \mathbbmss{T}_i$. Now, observe that
\begin{align}
    \label{Eq.Type_I_Error_1}
    \Pr\left( T(\fY(i),c_i) > \floor{n\delta_{\beta}} \right) & \stackrel{(a)}{=} \Pr \left( d_{\rm H}(\fY(i),\fc_i) > \floor{n\delta_{\beta}} \right)
    \nonumber\\
    & \stackrel{(b)}{=} \sum_{l = \floor{n\delta_{\beta}}+1}^{n} \binom{n}{l} \varepsilon^l (1 - \varepsilon)^{n-l}
    \,.
\end{align}
where $(a)$ follows by \eqref{Eq.Decoding_Metric} and $(b)$ holds by \eqref{Eq.Hamm_Distan_Dist}. In order to bound \eqref{Eq.Type_I_Error_1}, we proceed to apply the bound provided in \eqref{Ineq.Binom_Tail_2} given in Lemma~\ref{Lem.Bound_Sum_Tail_Binomial_Distribution_Part_2}: Observe that
\begin{align}
    \label{Ineq.Condition_Lemma_7_LB}
    \frac{l}{n} & = \frac{\floor{n\delta_{\beta}}+1}{n}
    \nonumber\\
    & \stackrel{(a)}{>} \frac{n\delta_{\beta}}{n}
    \nonumber\\
    & = \delta_{\beta}
    \nonumber\\
    & \stackrel{(b)}{>} \varepsilon \,,
\end{align}
where $(a)$ follows since $x < \floor{x} + 1$ for real $x$ and $(b)$ holds by \eqref{Ineq.delta_2}. On the other hand,
\begin{align}
    \label{Ineq.Condition_Lemma_7_UB}
    \frac{l}{n} & = \frac{\floor{n\delta_{\beta}}+1}{n}
    \nonumber\\
    & \leq \frac{\max \floor{n\delta_{\beta}} + 1}{n}
    \nonumber\\
    & \stackrel{(a)}{<} \frac{\floor{n \max \left( \varepsilon + \beta \left( \frac{1}{2} - \varepsilon \right) \right)} + 1}{n}
    \nonumber\\
    & \stackrel{(b)}{<} \frac{\floor{n/2} + 1}{n}
    \nonumber\\
    & \stackrel{n \geq 3}{<} 1
    \,,
\end{align}
where $(a)$ follows by \eqref{Ineq.delta_2} and $(b)$ holds since $\varepsilon + \beta \left( \frac{1}{2} - \varepsilon \right)$ is upper bounded by the boundary value of $\varepsilon$, i.e., where $\varepsilon = \frac{1}{2}$. Observe that the last inequality in \eqref{Ineq.Condition_Lemma_7_UB} holds for sufficiently large $n$. Now, since the inequalities provided in \eqref{Ineq.Condition_Lemma_7_LB} and \eqref{Ineq.Condition_Lemma_7_UB} fulfill the conditions in Lemma~\ref{Lem.Bound_Sum_Tail_Binomial_Distribution_Part_2}, we employ Lemma~\ref{Lem.Bound_Sum_Tail_Binomial_Distribution_Part_2} to establish the following lower bound on \eqref{Eq.Type_I_Error_1} as follows
\begin{align}
    \label{Eq.Type_I_Error_2}
    & \Pr\left( T(\fY(i),c_i) > \floor{n\delta_{\beta}} \right)
    \nonumber\\
    & = \sum_{l = \floor{n\delta_{\beta}}+1}^{n} \binom{n}{l} \varepsilon^l (1 - \varepsilon)^{n-l}
    \nonumber\\
    & \leq \left[ \frac{\left( \floor{n\delta_{\beta}} + 1 \right)(1-\varepsilon)}{\left( \floor{n\delta_{\beta}} + 1 \right) (1-\varepsilon) - \left[ n - (\floor{n\delta_{\beta}}+1)\right]\varepsilon} \right] \cdot 2^{-n\left[T_\varepsilon\left(\frac{\floor{n\delta_{\beta}}+1}{n}\right) - H\left(\frac{\floor{n\delta_{\beta}}+1}{n}\right)\right]} \,.
\end{align}
Observe that the denominator in \eqref{Eq.Type_I_Error_2} is always a strict positive term, since assuming we arrive to a trivial inequality as follows
\begin{align}
    \label{Ineq.Condition_Lemma_7_Verification}
    \left( \floor{n\delta_{\beta}} + 1 \right) (1 - \varepsilon) & > \left[ n - (\floor{n\delta_{\beta}}+1) \right]\varepsilon \Longleftrightarrow \\
    \floor{n\delta_{\beta}} + 1 - \varepsilon \floor{n\delta_{\beta}} - \varepsilon & > n \varepsilon - \varepsilon\floor{n\delta_{\beta}} - \varepsilon \Longleftrightarrow \\
    \floor{n\delta_{\beta}} + 1 & > n\varepsilon \Longleftrightarrow \\
    \frac{\floor{n\delta_{\beta}} + 1}{n} & > \varepsilon \,,
\end{align}
which is already verified in \eqref{Ineq.Condition_Lemma_7_LB}. Now, we proceed to find a simplified upper bound on the left hand side coefficient in the bracket given in \eqref{Eq.Type_I_Error_2} as follow
\begin{align}
    \label{Ineq.Coefficient_Simplified}
    & \frac{\left( \floor{n\delta_{\beta}} + 1 \right)(1-\varepsilon)}{\left( \floor{n\delta_{\beta}} + 1 \right) (1-\varepsilon) - \left[ n - (\floor{n\delta_{\beta}}+1)\right]\varepsilon}
    \nonumber\\
    & \stackrel{(a)}{=} \frac{\left( n\delta_{\beta} + 1 \right)(1-\varepsilon)}{\left( \floor{n\delta_{\beta}} + 1 \right) - \varepsilon \left( \floor{n\delta_{\beta}} + 1 \right) - n\varepsilon + \varepsilon \left( \floor{n\delta_{\beta}} + 1 \right)}
    \nonumber\\
    & \leq \frac{\left( n\delta_{\beta} + 1 \right)(1-\varepsilon)}{\left( \floor{n\delta_{\beta}} + 1 \right) - n\varepsilon}
    \nonumber\\
    & \stackrel{(b)}{\leq} \frac{\left( n\delta_{\beta} + 1 \right)(1-\varepsilon)}{n\delta_{\beta} - n\varepsilon} \,,
\end{align}
where $(a)$ holds by exploiting $x \leq \floor{x}$ for real $x$ and simplifying the denominator by distributing $\varepsilon$ over the bracket, and $(b)$ follows since
\begin{align}
    n\delta_{\beta} & < \floor{n\delta_{\beta}}+1 \Longleftrightarrow \\
    n\delta_{\beta} - n\varepsilon & < \floor{n\delta_{\beta}}+1 -n\varepsilon \Longleftrightarrow \\
    \frac{1}{n\delta_{\beta} - n\varepsilon} & > \frac{1}{\floor{n\delta_{\beta}}+1 - n\varepsilon} \,.
\end{align}
where the first inequality follows since $x < \floor{x} + 1$ for real $x$. Thereby, employing \eqref{Ineq.Coefficient_Simplified} unto \eqref{Eq.Type_I_Error_2}, we obtain
\begin{align}
     \label{Eq.Type_I_Error_3}
    \Pr\left( \left| T(\fY(i),c_i) \right| > \floor{n\delta_{\beta}} \right) & = \sum_{l = \floor{n\delta_{\beta}}+1}^{n} \binom{n}{l} \varepsilon^l (1 - \varepsilon)^{n-l}
    \nonumber\\
    & \leq \frac{\left( n\delta_{\beta} + 1 \right)(1-\varepsilon)}{n\delta_{\beta} - n\varepsilon} \cdot 2^{-n\left[T_\varepsilon\left(\frac{\floor{n\delta_{\beta}}+1}{n}\right) - H\left(\frac{\floor{n\delta_{\beta}}+1}{n}\right)\right]}
    \nonumber\\
    & = \frac{\left( \delta_{\beta} + \frac{1}{n} \right) \left( 1-\varepsilon \right)}{\delta_{\beta} - \varepsilon} \cdot 2^{-n\left[T_\varepsilon\left(\frac{\floor{n\delta_{\beta}}+1}{n}\right) - H\left(\frac{\floor{n\delta_{\beta}}+1}{n}\right)\right]}
    \nonumber\\
    & \triangleq \zeta_{1,n} \,.
\end{align}
Observe that the exponent of exponential term is always \emph{strictly} positive, since for $\varepsilon \in (0,\frac{1}{2})$, the arguments of $T_{\varepsilon}(.)$ and $H(.)$ are strictly less than $\frac{1}{2}$. That is, we have the following
\begin{align}
    T_\varepsilon\left(\frac{\floor{n\delta_{\beta}}+1}{n}\right) > H\left( \frac{\floor{n\delta_{\beta}}+1}{n} \right) \,,
\end{align}
The argument is as follows
\begin{align}
    \frac{l}{n} & = \frac{\floor{n\delta_{\beta}}+1}{n}
    \nonumber\\
    & \leq \frac{\max \floor{n\delta_{\beta}} + 1}{n}
    \nonumber\\
    & \stackrel{(a)}{<} \frac{\floor{n \max \left( \varepsilon + \beta \left( \frac{1}{2} - \varepsilon \right) \right)} + 1}{n}
    \nonumber\\
    & \stackrel{(b)}{<} \frac{\floor{n/2} + 1}{n}
    \nonumber\\
    & \stackrel{(c)}{\leq} \frac{n/2 + 1}{n}
    \nonumber\\
    & \stackrel{(c)}{\leq} \frac{n \left( \frac{1}{2} + 1/n \right)}{n}
    \,,
\end{align}
which is strictly less than $\frac{1}{2}$ in the asymptotic, i.e., as $n \to \infty$, where $(a)$ and $(b)$ follows by the same arguments given for \eqref{Ineq.Condition_Lemma_7_UB}, and $(c)$ follows since $\floor{x} \leq x$ for real $x$.

Therefore, the difference for the evaluation of $T_{\varepsilon}(.)$ and $H(.)$ is always a \emph{strict} positive value; see Figure~\ref{Fig.BSC_Channel_TR}. Hence, $P_{e,1}(i) \leq e_1 \;,\, \forall i \in \mathbbmss{T}_\mathbbmss{K}$ holds for sufficiently large $n$ and arbitrarily small $e_1 > 0$. Thereby, the type I error probability satisfies $P_{e,1}(i) \leq \zeta_{1,n} \leq e_1$. This complete the analysis for the type I error probability.

Next, we address type II errors, i.e., when $\fY(i) \in \mathbbmss{T}_\mathbbmss{K}$ while the transmitter sent $\fc_i$ with $i \notin \mathbbmss{K}$. Then, for each possible $\binom{M}{K}$ cases of $\mathbbmss{K}$, where $i \notin \mathbbmss{K}$, the type II error probability is given by
\begin{align}
    \label{Eq.Type_II}
    P_{e,2} \left( i,\mathbbmss{K} \right) & = \Pr \left( \fY(i) \in \mathbbmss{T}_\mathbbmss{K} \right)
    \nonumber\\
    & = \Pr \left( \fY(i) \in \bigcup_{j \in \mathbbmss{K}} \mathbbmss{T}_j \right)
    \nonumber\\
    & \stackrel{(a)}{=} \Pr\left( \bigcup_{j \in \mathbbmss{K}} \Big\{ T(\fY(i),c_j) \leq \floor{n\delta_{\beta}} \Big\} \right)
    \nonumber\\
    & \stackrel{(b)}{=} \Pr\left( \bigcup_{j \in \mathbbmss{K}} \Big\{ d_{\rm H}(\fY(i),\fc_j) \leq \floor{n\delta_{\beta}} \Big\} \right)
    \nonumber\\
    & \stackrel{(c)}{\leq}
    \sum_{j \in \mathbbmss{K}} \Pr\left( d_{\rm H}(\fY(i),\fc_j) \leq \floor{n\delta_{\beta}} \right)
    \nonumber\\
    & \leq K \cdot \Pr\left( d_{\rm H}(\fY(i),\fc_j) \leq \floor{n\delta_{\beta}} \right)
\end{align}
where $(a)$ follows by \eqref{Def.Dec_DKI}, $(b)$ holds by \eqref{Eq.Decoding_Metric} and $(c)$ follows by the \emph{union bound}, i.e., the probability of union of events is upper bounded by the sum of probability of the individual events. Let define the following events
\begin{align}
    \label{Eq.Event_1}
    \F_{\delta_{\beta}}(i) & \triangleq \left\{ \fY \in \mathbf{H}^n \;;\, d_{\rm H}(\fY(i),\fc_i) \leq \floor{n\delta_{\beta}} \right\} \,,
    \\
    \F_{\delta_{\beta}}(i,j) & \triangleq \left\{ \fY \in \mathbf{H}^n \;;\, d_{\rm H}(\fY(i),\fc_j) \leq \floor{n\delta_{\beta}} \right\} \,.
    \label{Eq.Event_2}
\end{align}
Next, employing the \emph{law of total probability} with respect to the event $\big\{ d_{\rm H}(\fY(i),\fc_i) \leq \floor{n\delta_{\beta}} \big\}$, we establish an upper bound on $\Pr\left( d_{\rm H}(\fY(i),\fc_j) \leq \floor{n\delta_{\beta}} \right)$ given in \eqref{Eq.Type_II} as follows
\begin{align}
    \label{Ineq.Type_II_Decompose}
    \Pr\left( d_{\rm H}(\fY(i),\fc_j) \leq \floor{n\delta_{\beta}} \right) & \stackrel{(a)}{=} \Pr\left( \F_{\delta_{\beta}}(i,j) \cap \F_{\delta_{\beta}}(i) \right) + \Pr\left( \F_{\delta_{\beta}}(i,j) \cap \F_{\delta_{\beta}}^c(i) \right)
    \nonumber\\
    & \stackrel{(b)}{\leq} \Pr\left( \F_{\delta_{\beta}}(i,j) \cap \F_{\delta_{\beta}}(i) \right) + \Pr\left( \F_{\delta_{\beta}}^c(i) \right)
    \nonumber\\
    & \stackrel{(c)}{=} \Pr\left( \F_{\delta_{\beta}}(i,j) \cap \F_i(\delta_{\beta}) \right) + \Pr\left( d_{\rm H}(\fY(i),\fc_i) > \floor{n\delta_{\beta}} \right)
    \nonumber\\
    & \stackrel{(d)}{\leq} \Pr\left( \F_{\delta_{\beta}}(i,j) \cap \F_{\delta_{\beta}}(i) \right) + \zeta_{1,n} \,,
\end{align}
where $(a)$ holds by the \emph{law of total probability}, $(b)$ follows since $\F_i^c(\delta_{\beta}) \supset \F_{\delta_{\beta}}(i,j) \cap \F_i^c(\delta_{\beta})$, $(c)$ holds by \eqref{Eq.Event_1}, and $(d)$ exploits \eqref{Eq.Type_I_Error_3}.

Now we focus on the event $\F_{\delta_{\beta}}(i,j) \cap \F_{\delta_{\beta}}(i)$. Let
\begin{align}
    d & \triangleq d_{\rm H}(\fc_{i},\fc_{j})
    \nonumber\\
    & \stackrel{(a)}{\geq} \floor{n\beta} + 1 \,.
\end{align}
where $(a)$ follows by the assumption made in the code construction regarding the minimum Hamming distance; see Lemma~\ref{Lem.Exhaustion_Case_1} and \eqref{Ineq.d_H_LB_2}. Now, without loss of generality, we may assume that the two sequence $\fc_i$ and $\fc_j$ differ in the first $d$ symbols, i.e.,
\begin{align}
    \fc_i & = \left( c_{i_1},c_{i_2},\ldots,c_{i_d},c_{i_{d+1}},\ldots,c_{i_n} \right) 
    \nonumber\\
    \fc_j & = \left( c_{j_1},c_{j_2},\ldots,c_{j_d},c_{j_{d+1}},\ldots,c_{j_n} \right)
    \nonumber\\
    \fy & = \left( y_1,y_2,\ldots,y_d,y_{d+1},\ldots,y_n \right) \,,
 \end{align}
where $\fy$ is the realization of vector $\fY(i)$. Therefore, the $n-d$ last symbols (bits) of $\fc_i$ and $\fc_j$ are identical. Observe that the event $\left\{ d_{\rm H}(\fY(i),\fc_i) \leq \floor{n\delta_{\beta}} \right\}$ implies that the received vector $\fy$ and $\fc_i$ differ in $p$ bits, where $p \leq \floor{n\delta_{\beta}}$, i.e.,
\begin{align}
    d_{\rm H}(\fy,\fc_i) = p \leq \floor{n\delta_{\beta}} \,.
\end{align}
Now, we assume that $p_1$ bits out of the $p$ bits happens in the first $d$ bits, i.e., 
\begin{align}
    d_{\rm H}(\fy|_1^d,\fc_i|_1^d) = p_1 \,,
\end{align}
where
\begin{align}
    \fc_i|_1^d & \triangleq \left( c_{i_1},c_{i_2},\ldots,c_{i_d} \right) 
    \nonumber\\
    \fy|_1^d & \triangleq \left( y_1,y_2,\ldots,y_d \right)
    \,,
\end{align}
and $p_2$ bits with $p_2 = p - p_1$ happens in last $n-d$ bits, i.e., 
\begin{align}
    d_{\rm H}(\fy|_{d+1}^n,\fc_i|_{d+1}^n) = p_2 \,,
\end{align}
where
\begin{align}
    \fc_i|_{d+1}^n & \triangleq \left( c_{i_{d+1}},\ldots,c_{i_n} \right)
    \nonumber\\
    \fy|_{d+1}^n & \triangleq \left( y_{d+1},\ldots,y_n \right)
    \,,
\end{align}
Observe that since the symbols of sequences are bits, i.e., either 0 or 1, therefore, $d = d_{\rm H}(\fc_i,\fc_j)$ implies that the two sequences $\fc_i$ and $\fc_j$ are complementary for the first $d$ bits. Now, we infer that if the two sequences $\fy|_1^d$ and $\fc_i|_1^d$ differ in $p_1$, then $\fy|_1^d$ and $\fc_i|_1^d$ are identical in those $p_1$ bits. Hence,
\begin{align}
    d_{\rm H}(\fy|_1^d,\fc_j|_1^d) = d - p_1 \,,
\end{align}
Now, if we collect all the positions for which $\fy|_1^n$ and $\fc_j|_1^n$ differ, we obtain
\begin{align}
    \label{Eq.Hamming_Dist_y_c_j}
    d_{\rm H}(\fy,\fc_j) & = d_{\rm H}(\fy|_1^n,\fc_j|_1^n)
    \nonumber\\
    & = d_{\rm H}(\fy|_1^d,\fc_j|_1^d) + d_{\rm H}(\fy|_{d+1}^n,\fc_j|_{d+1}^n)
    \nonumber\\
    & = d - p_1 + p_2 \,.
\end{align}
Observe that since we restrict ourselves to the event
\begin{align}
    \F_{\delta_{\beta}}(i,j) \cap \F_i^c(\delta_{\beta}) \triangleq \left\{ d_{\rm H}(\fY(i),\fc_j) \leq \floor{n\delta_{\beta}} \right\} \cap \left\{ d_{\rm H}(\fY(i),\fc_i) \leq \floor{n\delta_{\beta}} \right\} \,,
\end{align}
we deduce that $d_{\rm H}(\fy|_1^n,\fc_j|_1^n)$, therefore by \eqref{Eq.Hamming_Dist_y_c_j}, we obtain
\begin{align}
    \label{Ineq.p_2_UB_1}
    d - p_1 + p_2 & \leq \floor{n\delta_{\beta}} \Rightarrow
    \nonumber\\
    p_2 & \leq \floor{n\delta_{\beta}} - d + p_1 \,.
\end{align}
On the other hand, since $d_{\rm H}(\fy,\fc_j) \leq \floor{n\delta_{\beta}}$, we obtain
\begin{align}
     \label{Ineq.p_2_UB_2}
    p & \leq \floor{n\delta_{\beta}} \Rightarrow \\
    p_1 + p_2 & \leq \floor{n\delta_{\beta}} \Rightarrow \\
    p_2 & \leq \floor{n\delta_{\beta}} - p_1 \,.
\end{align}
Now in order to calculate $\Pr\left( d_{\rm H}(\fY(i),\fc_j) \leq \floor{n\delta_{\beta}} \right)$ in \eqref{Eq.Type_II}, we first fix $p_1$ and then sum up over all possible cases for the $p_2$, then we would have a second sum which runs for values of $p_1$ from 0 to $d$. Observe that the $p_2$ has two upper bounds given in \eqref{Ineq.p_2_UB_1} and \eqref{Ineq.p_2_UB_2}, therefore, in the calculation, we restrict ourselves to the minimum of those two upper bounds. Let define $p_2^{\rm UB} \triangleq \min\left\{\floor{n\delta_{\beta}}-p_1,\floor{n\delta_{\beta}}-d+p_1\right\}$. Thereby,
\begin{align}
    \label{Ineq.Type_II_1}
    & \Pr\left( \F_{\delta_{\beta}}(i,j) \cap \F_{\delta_{\beta}}(i) \right) 
    \nonumber\\
    & \stackrel{(a)}{\leq} \sum_{p_1=0}^d \binom{d}{p_1} \cdot \sum_{p_2=0}^{p_2^{\rm UB}} \binom{n-d}{p_2} \varepsilon^{p_1+p_2} (1-\varepsilon)^{n-\left(p_1+p_2\right)+d-d}
    \nonumber\\
    & \stackrel{(b)}{=} \left[ \sum_{p_1=0}^d \binom{d}{p_1} \varepsilon^{p_1} (1-\varepsilon)^{d-p_1} \right] \cdot \left[ \sum_{p_2=0}^{p_2^{\rm UB}} \binom{n-d}{p_2} \varepsilon^{p_2} (1-\varepsilon)^{n-d-p_2} \right] \,,
\end{align}
where $(a)$ holds since $p = p_1 + p_2$, and $(b)$ follows since every expression that is independent of the sum's variable can be shifted left behind the inner sum. In $(b)$, we have added $0 = d - d\,,$ to obtain the correct form for the two binomial distribution expressions. Now, observe that the first sum is the Binomial cumulative distribution function at point $x=d$ and can be upper bounded by 1, i.e.,
\begin{align}
    \label{Ineq.Type_II_First_Sum}
    \sum_{p_1=0}^d \binom{d}{p_1} \varepsilon^{p_1} (1-\varepsilon)^{d-p_1} & = \Pr \left( p_1 \leq d \right)
    \nonumber\\
    & = B_X(x)|_{x=d}
    \nonumber\\
    & = B_X(d)
    \nonumber\\
    & = 1 \,.
\end{align}
Now, let focus on the second sum in \eqref{Ineq.Type_II_1} for which we establish an upper bound by maximizing $p_2^{\rm UB}$ through setting $p_1= \floor{\frac{d}{2}}$, i.e.,
\begin{align}
    \underset{p_1}{\argmax} \; p_2^{\rm UB} = \floor{\frac{d}{2}}
\end{align}
\begin{align}
    \max p_2^{\rm UB} & \triangleq \max \min \left\{ \floor{n\delta_{\beta}}-p_1,\floor{n\delta_{\beta}}-d+p_1\right\}
    \nonumber\\
    & = \min \left\{ \floor{n\delta_{\beta}}-p_1,\floor{n\delta_{\beta}}-d+p_1\right\}\Big|_{p_1=\floor{\frac{d}{2}}}
    \nonumber\\
    & = \left\{ \floor{n\delta_{\beta}}-\floor{\frac{d}{2}},\floor{n\delta_{\beta}}-d+\floor{\frac{d}{2}}\right\}
    \nonumber\\
    & = \left\{ \floor{n\delta_{\beta}}-\floor{\frac{d}{2}},\floor{n\delta_{\beta}}- \left( d - \floor{\frac{d}{2}} \right) \right\}
    \nonumber\\
    & = \floor{n\delta_{\beta}}-d+\floor{\frac{d}{2}} \,,
\end{align}
where the last equality holds since by $\floor{\frac{d}{2}} \leq \frac{d}{2}$ for real $\frac{d}{2}$, we obtain $\frac{d}{2}\leq d-\floor{\frac{d}{2}}$.

Now, we exploit the inequality \eqref{Ineq.Binom_Tail_3} given in Lemma~\ref{Lem.Binom_Tail_3} to obtain an upper bound for the second sum in \eqref{Ineq.Type_II_1} as follows: First we check whether the required condition in Lemma~\ref{Lem.Binom_Tail_3} are satisfied or not. Namely, we set $k = \floor{n\delta_{\beta}}-d+\floor{\frac{d}{2}}$ and $n = n - d$. Now we calculate their ratio as follow
\begin{align}
    \label{Ineq.Condition_Lemma_8_1}
    \frac{k}{n-d} & = \frac{\floor{n\delta_{\beta}} - d + \floor{\frac{d}{2}}}{n - d}
    \nonumber\\
    & \stackrel{(a)}{\leq} \frac{n\delta_{\beta} - d + \frac{d}{2}}{n - d}
    \nonumber\\
    & = \frac{n\delta_{\beta} - \frac{d}{2}}{n - d}
    \nonumber\\
    & = \frac{\delta_{\beta} - \frac{d}{2n}}{1 - \frac{d}{n}}
    \nonumber\\
    & \stackrel{(b)}{<} \frac{\delta_{\beta} - \frac{\beta}{2}}{1 - \beta}
    \nonumber\\
    & \triangleq \tau
    \,,
\end{align}
where $(a)$ holds since $\floor{x} \leq x$ for real $x$ and $(b)$ holds by the following argument: We assume that $(b)$ holds and assuming that $\delta_{\beta} \neq \frac{1}{2}$, we resulted in a trivial inequality, namely, $d > n\beta$, i.e,
\begin{align}
    \frac{\delta_{\beta} - \frac{d}{2n}}{1 - \frac{d}{n}} & < \frac{\delta_{\beta} - \frac{\beta}{2}}{1 - \beta} \Rightarrow \\
    \left( \delta_{\beta} - \frac{d}{2n} \right) \left( 1 - \beta \right) & < \left( \delta_{\beta} - \frac{\beta}{2} \right) \left( 1 - \frac{d}{n} \right) \Rightarrow \\
    \delta_{\beta} - \beta\delta_{\beta} - \frac{d}{2n} + \frac{\beta d}{2n} & < \delta_{\beta} - \frac{\delta_{\beta} d}{n} - \frac{\beta}{2} + \frac{\beta d}{2n} \Rightarrow \\
    \beta \left( \frac{1}{2} - \delta_{\beta} \right) & < \frac{d}{2n} - \frac{\delta_{\beta} d}{n} \Rightarrow \\
    \beta \left( \frac{1}{2} - \delta_{\beta} \right) & < \frac{d}{n} \left( \frac{1}{2} - \delta_{\beta} \right) \Rightarrow \\
    n\beta & < d
    \,,
\end{align}
which can be deduced from the assumption made in the code construction given in \eqref{Ineq.d_H_LB_2} and \eqref{Ineq.d_H_LB_2}, i.e.,
\begin{align}
    d_{\rm H}(\fc_{i},\fc_{j}) & \geq \floor{n\beta} + 1
    \nonumber\\
    & \stackrel{(a)}{>} n\beta - 1 + 1
    \nonumber\\
    & = n\beta
    \;.
\end{align}
where $(a)$ holds since $\floor{n\beta} > n\beta - 1$ for real $n\beta$. Now, we exploit \eqref{Ineq.delta_2}, to show that \eqref{Ineq.Condition_Lemma_8_1} is upper bounded by $\varepsilon$ as follows
\begin{align}
    \label{Ineq.Condition_Lemma_8_2}
    \delta_{\beta} & < \varepsilon + \beta \left( \frac{1}{2} - \varepsilon \right) \Rightarrow
    \nonumber\\
    \delta_{\beta} & < \varepsilon + \frac{\beta}{2} - \beta \varepsilon \Rightarrow
    \nonumber\\
    \delta_{\beta} - \frac{\beta}{2} & < \varepsilon (1 - \beta) \Rightarrow
    \nonumber\\
    \frac{\delta_{\beta} - \frac{\beta}{2}}{1 - \beta} & < \varepsilon \,.
\end{align}
Thereby, we apply safely the Lemma~\ref{Lem.Binom_Tail_3} with parameters $j=p_2$, $k = p_2^{\rm UB} \triangleq \floor{n\delta_{\beta}}-d+\floor{\frac{d}{2}}$ and $n = n - d$, and obtain
\begin{align}
    \label{Ineq.Type_II_Error_Sum_2}
    \sum_{p_2=0}^{\floor{n\delta_{\beta}}-d+\floor{\frac{d}{2}}} \binom{n-d}{p_2} \varepsilon^{p_2} (1-\varepsilon)^{n-d-p_2} & \leq \frac{\varepsilon((n-d)-k)}{\varepsilon (n-d)-k} \cdot
    2^{n\left[ H(\frac{k}{n-d})-T_{\varepsilon}(\frac{k}{n-d}) \right]}
    \nonumber\\
    & \leq \frac{\varepsilon \left(1 - \frac{k}{n-d} \right)}{\varepsilon - \frac{k}{n-d}} \cdot
    2^{n\left[ H(\frac{k}{n-d}) - T_{\varepsilon}(\frac{k}{n-d}) \right]}
    \,.
\end{align}
Let focus on the coefficient in \eqref{Ineq.Type_II_Error_Sum_2}. In the following, assuming an upper bound for it, we arrive to a trivial inequality, therefore, the upper bound is valid.
\begin{align}
    \label{Ineq.Type_II_Error_Sum_2_Coefficient}
    \frac{\varepsilon \left(1 - \frac{k}{n-d} \right)}{\varepsilon - \frac{k}{n-d}} & < \frac{\varepsilon(1 - \tau)}{\varepsilon - \tau} \,.
\end{align}
Observe that \eqref{Ineq.Type_II_Error_Sum_2_Coefficient} yield the following chain of expressions:
\begin{align}
    \frac{1 - \frac{k}{n-d}}{\varepsilon - \frac{k}{n-d}} & < \frac{1 - \tau}{\varepsilon - \tau} \Rightarrow \\
    \varepsilon - \tau - \frac{k\varepsilon}{n-d} + \frac{k\tau}{n-d} & < \varepsilon - \frac{k}{n-d} - \varepsilon\tau + \frac{k\tau}{n-d}
    \Rightarrow \\
    - \tau - \frac{k\varepsilon}{n-d} & < - \frac{k}{n-d} - \varepsilon\tau \Rightarrow \\
    \frac{k}{n-d} \left( 1 - \varepsilon \right) & < \tau \left( 1 - \varepsilon \right) \Rightarrow \\
    \frac{k}{n-d} & < \varepsilon \,,
\end{align}
which is trivial since it is already proved in \ref{Ineq.Condition_Lemma_8_1}. Now, observe that for $0 < \frac{k}{n-d} < \tau < \varepsilon$, the following holds
\begin{align}
    \label{Ineq.UB_Exponent}
    H\left(\frac{k}{n-d}\right) - T_{\varepsilon}\left(\frac{k}{n-d}\right) < H(\tau) - T_{\varepsilon}(\tau) \,,
\end{align}
see Figure~\ref{Fig.Err_Exp}. Therefore, since $\tau$ always yield a smaller exponent, we obtain an upper bound on the sum in \eqref{Ineq.Type_II_Error_Sum_2} as follows
\begin{align}
    \sum_{p_2=0}^{\floor{n\delta_{\beta}}-d+\floor{\frac{d}{2}}} \binom{n-d}{p_2} \varepsilon^{p_2} (1-\varepsilon)^{n-d-p_2} & \leq \frac{\varepsilon((n-d)-k)}{\varepsilon (n-d)-k} \cdot
    2^{n\left[ H(\frac{k}{n-d})-T_{\varepsilon}(\frac{k}{n-d}) \right]}
    \nonumber\\&
    \stackrel{(a)}{<} \frac{\varepsilon(1 - \tau)}{\varepsilon - \tau} \cdot 2^{n\left[ H(\frac{k}{n-d}) - T_{\varepsilon}(\frac{k}{n-d}) \right]}
    \nonumber\\
    & \stackrel{(b)}{<} \frac{\varepsilon \left(1 - \frac{k}{n-d} \right)}{\varepsilon - \frac{k}{n-d}} \cdot
    2^{n\left[ H(\tau) - T_{\varepsilon}(\tau) \right]}
    \nonumber\\
    & \triangleq \zeta_{0,n} \,,
\end{align}
where $(a)$ exploits \eqref{Ineq.Type_II_Error_Sum_2_Coefficient} and $(b)$ follows by \eqref{Ineq.UB_Exponent}. Thereby, recalling \eqref{Ineq.Type_II_1} and employing \eqref{Ineq.Type_II_First_Sum}, we obtain
\begin{align}
    \label{Ineq.Type_II_Error_Sum_2_2}
    \Pr\left( \F_{\delta_{\beta}}(i,j) \cap \F_{\delta_{\beta}}(i) \right) & \leq 1 \cdot \sum_{j=0}^k \binom{n-d}{j} \varepsilon^j(1-\varepsilon)^{n-d-j}
    \nonumber\\
    & < \frac{\varepsilon(1 - \tau)}{\varepsilon - \tau} \cdot 2^{n\left[ H(\tau) - T_{\varepsilon}(\tau) \right]}
    \nonumber\\
    & \triangleq \zeta_{0,n} \,.
\end{align}
Hence, recalling \eqref{Eq.Type_II} and \eqref{Ineq.Type_II_Decompose} we obtain
\begin{align}
    \label{Ineq.Type_II_Final}
    & P_{e,2} \left( i,\mathbbmss{K} \right)
    \nonumber\\
    & \leq K \cdot \left[ \Pr\left( d_{\rm H}(\fY(i),\fc_j) \leq \floor{n\delta_{\beta}} \right) \right]
    \nonumber\\
    & \leq K \cdot \left[ \Pr\left( \F_{\delta_{\beta}}(i,j) \cap \F_{\delta_{\beta}}(i) \right) + \zeta_{1,n} \right]
    \nonumber\\
    & = K \cdot \left[
    \frac{\varepsilon(1 - \tau)}{\varepsilon - \tau} \cdot 2^{n\left[ H(\tau) - T_{\varepsilon}(\tau) \right]} + \frac{\left( \delta_{\beta} + \frac{1}{n} \right) \left( 1-\varepsilon \right)}{\delta_{\beta} - \varepsilon} \cdot 2^{-n\left[T_\varepsilon\left(\frac{\floor{n\delta_{\beta}}+1}{n}\right) - H\left(\frac{\floor{n\delta_{\beta}}+1}{n}\right)\right]} \right]
    \nonumber\\
    & \stackrel{(a)}{=} 2^{\kappa n} \cdot \left[ \frac{\varepsilon(1 - \tau)}{\varepsilon - \tau} \cdot 2^{-n \left[ T_{\varepsilon}(\tau) - H(\tau) \right]} + \frac{\left( \delta_{\beta} + \frac{1}{n} \right) \left( 1-\varepsilon \right)}{\delta_{\beta} - \varepsilon} \cdot 2^{-n\left[T_\varepsilon\left(\frac{\floor{n\delta_{\beta}}+1}{n}\right) - H\left(\frac{\floor{n\delta_{\beta}}+1}{n}\right)\right]} \right]
    \nonumber\\
    & = \frac{\varepsilon(1 - \tau)}{\varepsilon - \tau} \cdot 2^{-n \left[ T_{\varepsilon}(\tau) - H(\tau) - \kappa \right]} + \frac{\left( \delta_{\beta} + \frac{1}{n} \right) \left( 1-\varepsilon \right)}{\delta_{\beta} - \varepsilon} \cdot 2^{-n \left[ T_\varepsilon\left( \frac{\floor{n\delta_{\beta}}+1}{n}\right) - H\left( \frac{\floor{n\delta_{\beta}}+1}{n} \right) - \kappa \right]} \,,
\end{align}
which implies that both the exponential factors given in \eqref{Ineq.Type_II_Final} should yields strict positive exponents, that is, we obtain two separate upper bounds on the $\kappa$ as follows
\begin{align}
    \kappa < T_{\varepsilon}(\tau) - H(\tau) \qquad \text{  and  } \qquad \kappa < T_\varepsilon\left( \frac{\floor{n\delta_{\beta}}+1}{n}\right) - H\left( \frac{\floor{n\delta_{\beta}}+1}{n} \right) \,,
\end{align}
Therefore,
\begin{align}
    \label{Ineq.kappa_UB}
    \kappa < \min \left\{ T_{\varepsilon}(\tau) - H(\tau) ,\, T_\varepsilon\left( \frac{\floor{n\delta_{\beta}}+1}{n}\right) - H\left( \frac{\floor{n\delta_{\beta}}+1}{n} \right) \right\} \,,
\end{align}

Now we focus on the first argument in \eqref{Ineq.kappa_UB}. In order to find an asymptotic value for the argument $\tau$, we calculate the following
\begin{align}
    \label{Ineq.tau}
    \lim_{\beta \to 0} \tau & \stackrel{(a)}{=} \lim_{\beta \to 0} \frac{\delta_{\beta} - \frac{\beta}{2}}{1 - \beta}
    \nonumber\\&
    = \delta_{\beta}
    \,,
\end{align}
where $(a)$ holds by \eqref{Ineq.Condition_Lemma_8_1}. Thereby,
\begin{align}
    \lim_{n\to\infty} T_{\varepsilon}(\tau) - H(\tau) = T_\varepsilon\left( \delta_{\beta} \right) - H\left( \delta_{\beta} \right)
    \,,
\end{align}
Now we focus on the second argument in \eqref{Ineq.kappa_UB} and provide the following asymptotic behavior:
\begin{align}
    \label{Ineq.kappa_UB_Final_2}
    \lim_{n\to\infty} T_\varepsilon\left( \frac{\floor{n\delta_{\beta}}+1}{n}\right) - H\left(\frac{\floor{n\delta_{\beta}}+1}{n} \right) & = T_\varepsilon\left( \lim_{n\to\infty} \frac{\floor{n\delta_{\beta}}+1}{n}\right) - H\left( \lim_{n\to\infty} \frac{\floor{n\delta_{\beta}}+1}{n} \right)
    \,,
\end{align}
where the equality holds since $T_{\varepsilon}(.)$ and $H(.)$ are continuous functions of $\delta_{\beta}$. Now, observe that since $\floor{n\delta_{\beta}} - 1< \floor{n\delta_{\beta}} \leq n\delta_{\beta}$ for real $n\delta_{\beta}$, we obtain
\begin{align}
    \lim_{n\to\infty} \frac{n\delta_{\beta} - 1 +1}{n} \leq \lim_{n\to\infty} \frac{\floor{n\delta_{\beta}}+1}{n} \leq \lim_{n\to\infty} \frac{n\delta_{\beta}+1}{n} \Rightarrow \nonumber\\
    \delta_{\beta} \leq \lim_{n\to\infty} \frac{\floor{n\delta_{\beta}}+1}{n} \leq \lim_{n\to\infty} \delta_{\beta} + \frac{1}{n} \stackrel{(a)}{\Rightarrow} \nonumber\\
    \lim_{n\to\infty} \frac{\floor{n\delta_{\beta}}+1}{n} = \delta_{\beta}
    \,.
\end{align}
where $(a)$ holds by the \emph{squeeze theorem}. Thereby,
\begin{align}
     \lim_{n\to\infty} T_\varepsilon\left( \frac{\floor{n\delta_{\beta}}+1}{n}\right) - H\left(\frac{\floor{n\delta_{\beta}}+1}{n} \right) = T_\varepsilon\left( \delta_{\beta} \right) - H\left( \delta_{\beta} \right)
    \,.
\end{align}
Therefore, recalling \eqref{Ineq.kappa_UB}, we obtain the following upper bound on the target identification rate $\kappa$:
\begin{align}
    \label{Ineq.kappa_UB_Final_3}
    \kappa & < \min \left\{ T_{\varepsilon}(\tau) - H(\tau) ,\, T_\varepsilon\left( \frac{\floor{n\delta_{\beta}}+1}{n}\right) - H\left( \frac{\floor{n\delta_{\beta}}+1}{n} \right) \right\}
    \nonumber\\
    & = \min \left\{ T_\varepsilon\left( \delta_{\beta} \right) - H\left( \delta_{\beta} \right) ,\, T_\varepsilon\left( \delta_{\beta} \right) - H\left( \delta_{\beta} \right)  \right\}
    \nonumber\\
    & = T_\varepsilon\left( \delta_{\beta} \right) - H\left( \delta_{\beta} \right)
    \,,
\end{align}
where the equality holds since $T_{\varepsilon}(.)$ and $H(.)$ are continuous functions of $\delta_{\beta}$.
Therefore, recalling \eqref{Ineq.Type_II_Final}, we obtain
\begin{align}
    P_{e,2}(i,j) & \leq \Pr\left( \F_{\delta_{\beta}}(i,j) \cap \F_{\delta_{\beta}}(i) \right) + \Pr\left( d_{\rm H}(\fY(i),\fc_i) > \floor{n\delta_{\beta}} \right)
    \nonumber\\
    & \leq \zeta_{0,n} + \zeta_{1,n}
    \nonumber\\
    & \leq \zeta_0 + \zeta_1
    \nonumber\\
    & \leq e_2 \,,\,
\end{align}
hence, $P_{e,2}(i,j) \leq e_2$ holds for sufficiently large $n$ and arbitrarily small $e_2 > 0$. We have thus shown that for every $e_1,e_2>0$ and sufficiently large $n$, there exists an $(n, M(n,R), K(n,\kappa), \allowbreak e_1, e_2)$ code.
\subsection{Upper Bound (Converse Proof)}
\label{Sec.Conv}
The converse proof consists of employing the following lemma on the size of a DKI code. In particular, depending on whether or not a Hamming weight constraint is present, we divide in two cases and address them separately. More specifically, we use the following observation. Let $R > 0$ be a DKI achievable rate. We assume to the contrary that there exist two distinct messages $i_1$ and $i_2$ that are represented by a common codeword, i.e., $\fc_{i_1} = \fc_{i_2} = x^n$, and show that this assumption result in a contradiction, namely, the sum of type I and type II error probabilities converges to one from left, i.e.,
\begin{align}
    \lim_{n\to\infty} P_{e,1}(i_1) + P_{e,2}(i_2,\mathbbmss{K}) = 1 \;.
\end{align}
Hence our assumption is false and the number of messages $2^{nR}$ is bounded by either the size of the subset of the input sequences that satisfy the input constraint or the entire input space.
\begin{lemma}
\label{Lem.DConverse}
Consider a sequence of $(n,\allowbreak M(n\allowbreak,R),\allowbreak K(n,\allowbreak \kappa), \allowbreak e_1^{(n)}, \allowbreak e_2^{(n)})$ codes $(\C^{(n)},\sT^{(n)})$ such that $e_1^{(n)}$ and $e_2^{(n)}$ tend to zero as $n\rightarrow\infty$. Then, given a sufficiently large $n$, the codebook $\C^{(n)}$ satisfies the following property: There cannot be two distinct messages $i_1,i_2 \in [\![M]\!]$ that are represented by the same codeword, i.e.,
\begin{align}
    i_1 \neq i_2 \; \quad \Rightarrow\quad \fc_{i_1} \neq \fc_{i_2} \;.
\end{align}
\end{lemma}
\begin{proof}
Assume to the contrary that there exist two messages $i_1$ and $i_2$, where $i_1 \neq i_2$, such that
\begin{align}
    \fc_{i_1} = \fc_{i_1} = x^n \;,
\end{align}
for some $x^n \in \X^n$. Since $(\C^{(n)},\sT^{(n)})$ forms a $(n,\allowbreak M(n\allowbreak,R),\allowbreak K(n,\allowbreak \kappa), \allowbreak e_1^{(n)}, \allowbreak e_2^{(n)})$ code, it implies that for every possible arrangement of $\big\{ \mathbbmss{K}, \mathbbmss{K}^c \big\}$, $e_1$ and $e_2$ tends to zero. Therefore, the existence of a \emph{desired} arrangement of $\big\{ \mathbbmss{K}, \mathbbmss{K}^c \big\}$ where $\mathbbmss{K} \subseteq [\![M]\!]$ with the property that $i_1 \in \mathbbmss{K}$ and $i_2 \in \mathbbmss{K}^c$, is guaranteed. Thereby, we obtain
\begin{align}
    P_{e,1}(i_1) = W^n(\sT_{\mathbbmss{K}}^c \,|\, x^n = \fc_{i_1})_{i_1 \in \mathbbmss{K}} & \leq e_1^{(n)} \;,
    \nonumber\\
    P_{e,2}(i_2,\mathbbmss{K}) = W^n(\sT_{\mathbbmss{K}} \,|\, x^n = \fc_{i_2})_{i_2 \notin \mathbbmss{K}} & \leq e_2^{(n)} \;.\,
\end{align}
This leads to a \emph{contradiction} since
\begin{align}
    1 & = W^n(\sT_{\mathbbmss{K}}^c \,|\, x^n) + W^n(\sT_{\mathbbmss{K}} \,|\, x^n)
    \nonumber\\
    & = P_{e,1}(i_1) + P_{e,2}(i_2,\mathbbmss{K})
    \nonumber\\ 
    & \leq e_1^{(n)} + e_2^{(n)} \;,
\end{align}
where the last inequality exploits the definition of type I/II error probabilities given in \eqref{Eq.TypeI-Error-DKI} and \eqref{Eq.TypeII-Error-DKI}. Hence, the assumption is false, and distinct messages $i_1$ and $i_2$ cannot share the same codeword.
\end{proof}
$\rhd$ \textbf{\textcolor{mycolor12}{Case 1 - With Hamming Weight Constraint ($0 < A < 1$):}}
By Lemma~\ref{Lem.DConverse}, each message has a distinct codeword. Hence, the number of messages is bounded by the number of input sequences that satisfy the input constraint. We divide in two cases, namely, where $0 < A < \frac{1}{2}$ and $\frac{1}{2} \leq A < 1$. For the first case, we obtain the following upper bound on the size of the DKI codebook:
\begin{align}
    \label{Ineq.Ineq.Converse_Case_1_0}
    2^{nR} & \leq \left| \B_{\f0}(n,nA) \right|
    \nonumber\\&
    = \left| \left\{ \fx \in \mathbf{H}^n : 0 \leq \sum_{t=1}^n x_t \leq nA \right\} \right|
    \nonumber\\&
    \stackrel{(a)}{\leq} 2^{nH(A)}
    \;,
\end{align}
where $(a)$ exploits the upper bound on the volume of the Hamming ball provided in Lemma~\ref{Lem.Bound_Volume_Hamming_ball} for $0 < A < 1/2$. Thereby, \eqref{Ineq.Ineq.Converse_Case_1_0} implies
\begin{align}
    \label{Ineq.Converse_Case_1}
    R \leq H(A) \,.
\end{align}
Now, we proceed to calculate the upper bound on the size of the DKI codebook where $1/2 \leq A < 1$. We argue that this case is equivalent to having a Hamming weight constraint of the form $A^* = 1/2$. That is, the codewords with constraint $\sum_{t=1}^n x_t \leq nA^*$ where $A^* = 1/2$ fulfilled the same constraint with $\frac{1}{2} \leq A < 1$. The new Bernoulli input process has $1/2$ success probability, i.e., $X \sim \text{Bern}(1/2)$. Therefore, again employing Lemma~\ref{Lem.Bound_Volume_Hamming_ball} for the critical point $\varepsilon = 1/2$, we obtain
\begin{align}
    \label{Ineq.Ineq.Converse_Case_1_2_Comprehensive}
    2^{nR} & \leq \left| \B_{\f0}(n,nA^*) \right|
    \nonumber\\&
    = \left| \left\{ \fx \in \mathbf{H}^n : 0 \leq \sum_{t=1}^n x_t \leq nA^* \right\} \right|
    \nonumber\\&
    \leq 2^{nH\left( A^* = 1/2 \right)}
    \;,
\end{align}
which implies
\begin{align}
    \label{Ineq.Ineq.Converse_Case_1_2}
    R \leq H\left( A^* = 1/2 \right) = 1 \,.
\end{align}

$\rhd$ \textbf{\textcolor{mycolor12}{Case 2 - Without Hamming Weight Constraint ( $A \geq 1$):}}
In this case, the number of messages is bounded by the number of the input sequences, that is, the size of entire input space, i.e., $|\X|^n$. Therefore, we can establish the following upper bound on the size of the DKI codebook $2^{nR} \leq \left| \X \right|^n$ which for $|\X| = 2$ implies
\begin{align}
    \label{Ineq.Converse_Case_2}
    R & \leq \frac{1}{n} \log \left| \X \right|^n
    \nonumber\\
    & = 1 \,.\,
\end{align}
This completes the proof of Lemma~\ref{Th.DKI-Capacity}.

Thus, by \eqref{Ineq.Converse_Case_1}, \eqref{Ineq.Ineq.Converse_Case_1_2} and \eqref{Ineq.Converse_Case_2}, and exploiting the fact that DKI capacity is supremum of all achievable rate, the DKI coding rate is upper bounded by
\begin{align}
    \mathbb{C}_{\rm DKI} \left( \W_{\varepsilon},K \right) & \leq \begin{cases}
    H(A) & \text{if $0 < A < 1/2$}\\
    1 & \text{if $A \geq 1/2$} \,,
    \end{cases}
\end{align}
which completes the proof of Theorem~\ref{Th.DKI-Capacity}.
\section{Summary and Future Directions}
\label{Sec.Summary}
In this work, we studied the DKI problem over the binary symmetric channel. We assume that the transmitter is subject to a Hamming weight constraint $A$. Our results in this work may serve as a model for event-triggered based tasks in the context of future XG applications. In particular, we obtained the DKI capacity of the BSC with the codebook size of $M(n,R) = 2^{nR}$ equals to the entropy of the Hamming constraint value, i.e., $H(A)$. Our results for the DKI capacity of the BSC revealed that the conventional exponential scale of $2^{nR}$ which is used in the standard message transmission setting, is again the appropriate scale for codebook size. Further, we find out that the BSC features an \emph{exponentially} large set of target messages set, in the codeword length, $n$, i.e., $2^{\kappa n}$; and characterize all the possible valid range on the DKI target rate $\kappa$ which depending on the value of channels statistic may varies.

We show the achievability proof using a Hamming distance decoder and employing packing arrangement of hyper balls in the same line of arguments as is conducted for the basic \emph{Gilbert bound} method. In particular, in the presence of a Hamming weight constraint $A$, we pack hyper balls with radius $\floor{n\beta}$, inside a larger Hamming hyper ball, which results in $\sim 2^{nH(A)}$ codewords.

For the converse part, a similar approach as our previous work for the DMC \cite{Salariseddigh_IT,Salariseddigh_ICC,Salariseddigh_arXiv_ICC} is followed. That is, for the case where a non-trivial Hamming weight constraint is present ($0 < A < 1$), we establish an injection (one-to-one mapping) between the message set and the subset that is induced by the Hamming weight constraint. However, here we exploit the impact of generalized type I and type II error probability definitions with respect to the set of the target messages $\mathbbmss{K}$ in the course of the proof, namely, for any two distinct messages, there exist an arrangement of $\{ \mathbbmss{K},\mathbbmss{K}^c \}$ into which the two messages are categorized. In particular, we exploit the method of \emph{proof by the contradiction}. Namely, we first assume that two generic different messages $i_1$ and $i_2$ share the common codewords, and then show that such an assumption leads to a contradiction regarding the sum of the error probabilities, i.e., we derive that the sum of type I and type II error probabilities converges to one from left. Hence the falsehood of the early assumption is guaranteed. Therefore, the total number of messages $M=2^{nR}$ is bounded by the size of the induced subset, i.e., $M \leq 2^{nH(A)}$. For the case where $A \geq 1$, that is, in the absent of a Hamming weight constraint, a similar line of argument can be applied in order to establish the injective function.
 
Our presented results can be extended in several directions. Some of the possible topics for the future research are as follows:
\begin{itemize}
    \item \textcolor{blau_2b}{\textbf{Explicit Code Construction}}: Systematic design and explicit construction of the DKI codes and developing optimum (low complexity) encoding\,/\,decoding schemes for practical designs.
    \item \textcolor{blau_2b}{\textbf{Multi User}}: The extension of current work to more advanced settings, such as the multiple-input multiple-output or the multi-user models may seems more relevant in the applications of the future generations of wireless networks (XG).
    \item \textcolor{blau_2b}{\textbf{Fekete's Lemma}}: To analyse whether the pessimistic ($\underline{C} = \liminf_{n \to \infty} \allowbreak \frac{\log M(n,R)}{n \log n}$) and optimistic ($\overline{C} = \limsup_{n \to \infty} \frac{\log M(n,R)}{n \log n}$) capacities \cite{A06} are identical or not; That is, one can investigate the behavior of DKI capacity with respect to the so called \emph{Fekete's Lemma} \cite{Boche20}. See \cite{Boche20} for further discussions and details.
    \item \textcolor{blau_2b}{\textbf{Channel Reliability Function}}: Studying the behavior of the decoding errors (type I and type II) aiming at characterizing them as a function of the codeword length $n$, in the asymptotic, i.e., where $n \to \infty$, for the entire region of DKI coding rate $R$, i.e., $0 < \forall R < C$, demand us to know the so called \emph{channel reliability function (CRF)} \cite{Boche21}. However, to the best of the authors’ knowledge, such a function (CRF) for the DKI problem has not been studied so far in the literature.
\end{itemize}
\section*{Acknowledgements}

M.~J.~Salariseddigh was supported by the German Federal Ministry of Education and Research (BMBF) within the 6G-Life project, under Grant 16KISK002. We thank the German Research Foundation (DFG) within the Gottfried Wilhelm Leibniz Prize under Grant BO 1734/20-1 for their support of H. Boche. Thanks also go to the BMBF within the national initiative for “Post Shannon Communication (NewCom)” with the project “Basics, Simulation and Demonstration For New Communication Models” under Grant 16KIS1003K for their support of H. Boche and with the project “Coding Theory and Coding Methods For New Communication Models” under Grant 16KIS1005 for their support of C. Deppe. Further, we thank the DFG within Germany’s Excellence Strategy EXC-2111—390814868 and EXC-2092 CASA - 390781972 for their support of H. Boche. C. Deppe was supported in part by the BMBF under Grant 16KIS1005 and by the DFG Project DE1915/2-1. The authors acknowledge the financial support by the BMBF in the program of “Souverän. Digital. Vernetzt.”, joint project 6G-Life, project identification number: 16KISK002.
\appendix
\section{Lower Bound on The Volume of The Hamming ball}
\begin{lemma}
\label{Lem.LB_Volume_Hamming_ball}
Let $n,q \geq 2$ be positive integers and assume a real $p$ where $0 \leq \floor{n\varepsilon}/n \leq 1 - 1/q \,.$ Then, volume of the Hamming ball in the $q$-ary alphabet is lower bounded as follows 
\begin{align}
    \label{Ineq.Bound_Volume_Hamming_ball_LB}
    \text{Vol}\left( \B_{\mathbf{x}_0}(n,r) \right) \triangleq \sum_{j=0}^{\floor{n\varepsilon}} \binom{n}{j} (q-1)^j \geq q^{H_q\left(\frac{\floor{n\varepsilon}}{n}\right) - o\left( \log_q n \right)} \,.
\end{align}
\begin{proof}
    Observe that the \textbf{Stirling's approximation} \cite{Robbins55} gives the following double bound on $n!$
    \begin{align}
        \sqrt{2n\pi} \left( \frac{n}{e} \right)^n e^{\lambda_1(n)} \leq n! \leq \sqrt{2n\pi} \left( \frac{n}{e} \right)^n e^{\lambda_2(n)} \,.
    \end{align}
    Now, we have
    \begin{align}
        & \binom{n}{\floor{n\varepsilon}}
        \nonumber\\
        & = \frac{n!}{\floor{n\varepsilon}! \left( n - \floor{n\varepsilon} \right)!}
        \nonumber\\
        & > \frac{\sqrt{2n\pi} \cdot (\frac{n}{e})^n \cdot e^{\lambda_1(n)}}{\left[ \sqrt{2\floor{n\varepsilon} \pi} \cdot \left( \frac{\floor{n\varepsilon}}{e} \right)^{\floor{n\varepsilon}} \cdot e^{\lambda_1(n)} \right] \left[ \sqrt{2\left( n \left( 1 - \frac{\floor{n\varepsilon}}{n} \right) \right) \pi} \cdot \left( \frac{n \left( 1 - \frac{\floor{n\varepsilon}}{n} \right)}{e}\right)^{n\left(1 - \frac{\floor{n\varepsilon}}{n}\right)} \right]}
        \nonumber\\
        & = \left[ \frac{(\frac{n}{e})^n}{\left( \frac{\floor{n\varepsilon}}{e} \right)^{\floor{n\varepsilon}} \cdot \left( \frac{n \left( 1 - \frac{\floor{n\varepsilon}}{n} \right)}{e}\right)^{n(1 - \frac{\floor{n\varepsilon}}{n})} } \right] \cdot \left[ \frac{e^{\lambda_1(n) - \lambda_2\left(\floor{n\varepsilon}\right) - \lambda_2\left( n\left(1 - \frac{\floor{n\varepsilon}}{n}\right) \right)}}{\sqrt{2\pi \floor{n\varepsilon} \left( 1 - \frac{\floor{n\varepsilon}}{n} \right)}} \right]
        \nonumber\\
        & \stackrel{(a)}{=} \left[ \frac{1}{\left( \frac{\floor{n\varepsilon}}{n} \right)^{\floor{n\varepsilon}} \cdot \left( 1 - \frac{\floor{n\varepsilon}}{n} \right)^{n\left(1 - \frac{\floor{n\varepsilon}}{n}\right)} } \right] \cdot \left[ \frac{e^{\floor{n\varepsilon}} \cdot e^{n\left(1 - \frac{\floor{n\varepsilon}}{n}\right)}}{e^n} \right] \cdot \text{Res}(n)
        \nonumber\\
        & \stackrel{(b)}{=} \frac{\text{Res}(n)}{\left( \frac{\floor{n\varepsilon}}{n} \right)^{\floor{n\varepsilon}} \cdot \left( 1 - \frac{\floor{n\varepsilon}}{n} \right)^{n\left(1 - \frac{\floor{n\varepsilon}}{n}\right)} }
    \end{align}
    where $(a)$ holds since we let
    \begin{align}
        \text{Res}(n) \triangleq \frac{e^{\lambda_1(n) - \lambda_2\left(\floor{n\varepsilon}\right) - \lambda_2\left( n\left(1 - \frac{\floor{n\varepsilon}}{n}\right) \right)}}{\sqrt{2\pi \floor{n\varepsilon} \left( 1 - \frac{\floor{n\varepsilon}}{n} \right)}} \,,
    \end{align}
    and $(b)$ holds since
    \begin{align}
        \frac{e^{\floor{n\varepsilon}} \cdot e^{n\left(1 - \frac{\floor{n\varepsilon}}{n}\right)}}{e^n} = 1 \,.
    \end{align}
    Next, we proceed to bound the Hamming ball as follows: Observe that the volume of Hamming ball as provided in \eqref{Ineq.Bound_Volume_Hamming_ball_LB} is lower bounded by the Binomial coefficient for the largest index, i.e., $j = \floor{n\varepsilon}$. Therefore,
    \begin{align}
        \text{Vol}\left( \B_{\mathbf{x}_0}(n,r) \right) & \triangleq \sum_{j=0}^{\floor{n\varepsilon}} \binom{n}{j} (q-1)^j
        \nonumber\\
        & \geq \binom{n}{\floor{n\varepsilon}} (q-1)^{\floor{n\varepsilon}}
        \nonumber\\
        & > \frac{(q-1)^{\floor{n\varepsilon}}}{\left( \frac{\floor{n\varepsilon}}{n} \right)^{\floor{n\varepsilon}} \cdot \left( 1 - \frac{\floor{n\varepsilon}}{n} \right)^{n\left(1 - \frac{\floor{n\varepsilon}}{n}\right)}} \cdot \text{Res}(n)
        \nonumber\\
        & = q^{\log_q \left( \frac{(q-1)^{\floor{n\varepsilon}}}{\left( \frac{\floor{n\varepsilon}}{n} \right)^{\floor{n\varepsilon}} \cdot \left( 1 - \frac{\floor{n\varepsilon}}{n} \right)^{n\left(1 - \frac{\floor{n\varepsilon}}{n}\right)}} \right) + \log_q \text{Res}(n)}
        \nonumber\\
        & = q^{\floor{n\varepsilon} \log_q (q-1) - \floor{n\varepsilon} \log_q \frac{\floor{n\varepsilon}}{n} - n\left(1 - \frac{\floor{n\varepsilon}}{n}\right) \log_q\left( 1 - \frac{\floor{n\varepsilon}}{n} \right) + \log_q \text{Res}(n)}
        \nonumber\\
        & = q^{n\left( \frac{\floor{n\varepsilon}}{n} \log_q (q-1) - \frac{\floor{n\varepsilon}}{n} \log_q \frac{\floor{n\varepsilon}}{n} - \left(1 - \frac{\floor{n\varepsilon}}{n}\right) \log_q\left( 1 - \frac{\floor{n\varepsilon}}{n} \right) + \log_q \text{Res}(n) \right)}
        \nonumber\\
        & = q^{nH_q\left( \frac{\floor{n\varepsilon}}{n} \right) + \log_q \text{Res}(n)}
        \,.
    \end{align}
    Now by letting $\lambda_1(n) = 0$ and $\lambda_2(n)= 1 /(12n)$, we obtain
    \begin{align}
        \text{Res}(n) & = \frac{e^{- \frac{1}{12\floor{n\varepsilon}} - \frac{1}{n - \floor{n\varepsilon}}}}{\sqrt{2\pi \floor{n\varepsilon} \left( 1 - \frac{\floor{n\varepsilon}}{n} \right)}}
        \nonumber\\
        & \stackrel{(a)}{\leq} \frac{e^{- \frac{1}{12\floor{n\varepsilon}} - \frac{1}{n - \floor{n\varepsilon}}}}{\sqrt{2\pi \floor{n\varepsilon} \left( 1 - \varepsilon \right)}}
        \nonumber\\
        & \stackrel{(b)}{=} K(\varepsilon) \floor{n\varepsilon}^{- \frac{1}{2}} e^{- \frac{1}{12\floor{n\varepsilon}} - \frac{1}{n - \floor{n\varepsilon}}}
        \,,
    \end{align}
    where $(b)$ follows for sufficiently large $n$, since $\floor{n\varepsilon}$ $(b)$ holds for $K(\varepsilon) = \frac{1}{\sqrt{2\pi(1- \varepsilon)}}$. Therefore,
    \begin{align}
        \log_q \text{Res}(n) & =  \log_q K(\varepsilon) - \frac{1}{2} \log_q \floor{n\varepsilon} - \frac{1}{12\floor{n\varepsilon}} - \frac{1}{n - \floor{n\varepsilon}}
        \nonumber\\
        & = o\left( \log_q n \right) \,,
    \end{align}
    which implies that
    \begin{align}
        \lim_{n\to\infty} \frac{\log_q \text{Res}(n)}{\log_q n} = 0 \,.
    \end{align}
    Thereby,
    \begin{align}
        \text{Vol}\left( \B_{\mathbf{x}_0}(n,r) \right) & \triangleq \sum_{j=0}^{\floor{n\varepsilon}} \binom{n}{j} (q-1)^j
        \nonumber\\
        & \geq q^{nH_q\left( \frac{\floor{n\varepsilon}}{n} \right) + o\left( \log_q n \right)}
        \,.
    \end{align}
\end{proof}
\end{lemma}
\section{Upper Bound on The Volume of The Hamming ball}
\begin{lemma}[see {\cite[Lem.~$16.19$]{FG06}}]
\label{Lem.Bound_Volume_Hamming_ball}
Let integer $n \geq 1$ and $0 < \varepsilon \leq \frac{1}{2}$ with $n > \floor{n\varepsilon} \geq 1$. Then, volume of the Hamming ball in the binary alphabet is upper bounded as follows
\begin{align}
    \label{Ineq.Bound_Volume_Hamming_ball_UB}
    \text{Vol}\left( \B_{\mathbf{x}_0}(n,r) \right) \triangleq \sum_{j=0}^{\floor{n\varepsilon}}\binom{n}{j} \leq 2^{nH(\varepsilon)} \,,
\end{align}
\end{lemma}
\section{Bound on The Upper Tail of The Binomial Cumulative Distribution Function - Part 1}
\begin{lemma} [see {\cite[Probl.~$5.8-(c)$]{RG68}}]
\label{Lem.Bound_Sum_Tail_Binomial_Distribution_Part_1}
Let $0 < \varepsilon < 1$ and $\varepsilon < \frac{k}{n} < 1$. Then,
\begin{align}
    \label{Ineq.Binom_Tail_1}
    \binom{n}{k} \varepsilon^j(1-\varepsilon)^{n-k} \leq \sum_{j=k}^n \binom{n}{j} \varepsilon^j(1-\varepsilon)^{n-j} \leq \binom{n}{k} \varepsilon^k(1-\varepsilon)^{n-k} \left[ \frac{k(1-\varepsilon)}{k(1-\varepsilon)-(n-k)\varepsilon} \right] \,.
\end{align}
\end{lemma}
\begin{proof}
    See the provided hints in \cite[P.~531]{RG68} for further details.
\end{proof}
\section{Bound on The Upper Tail of The Binomial Cumulative Distribution Function - Part 2}
\begin{lemma}
\label{Lem.Bound_Sum_Tail_Binomial_Distribution_Part_2}
Let $0 < \varepsilon < 1$ and $\varepsilon < \frac{k}{n} < 1$. Then,
\begin{align}
    \label{Ineq.Binom_Tail_2}
    \sum_{j=k}^n \binom{n}{j} \varepsilon^j(1-\varepsilon)^{n-j} \leq 2^{n\left[ H(\frac{k}{n})-T_{\varepsilon}(\frac{k}{n}) \right]} \left[ \frac{k(1-\varepsilon)}{k(1-\varepsilon)-(n-k)\varepsilon} \right] 
    \,.
\end{align}
\begin{proof}
Recall that the equation of the tangent line to the binary entropy function $H(\delta_{\beta})$ at the specific point $\delta_{\beta} = \varepsilon$ is given by
\begin{align}
  & T_{\varepsilon}(\delta_{\beta})
  \nonumber\\
  & \stackrel{(a)}{=} H(\varepsilon) + (\delta_{\beta} - \varepsilon) \diff{H(\delta_{\beta})}{\delta_{\beta}}\Big|_{\delta_{\beta}=\varepsilon}
  \nonumber\\
  & \stackrel{(b)}{=} H(\varepsilon) + (\delta_{\beta} - \varepsilon) \log \left(\frac{1-\varepsilon}{\varepsilon}\right)
  \nonumber\\
  & = H(\varepsilon) + (\delta_{\beta} - \varepsilon) \left[ \log (1-\varepsilon) - \log \varepsilon \right]
  \nonumber\\
  & \stackrel{(c)}{=} - \varepsilon \log \varepsilon - (1 - \varepsilon) \log (1 - \varepsilon) + \delta_{\beta} \log (1-\varepsilon) - \delta_{\beta} \log \varepsilon - \varepsilon \log (1-\varepsilon) + \varepsilon \log \varepsilon
  \nonumber\\
  & = - \varepsilon \log \varepsilon - \log (1 - \varepsilon) + \varepsilon \log (1 - \varepsilon) + \delta_{\beta} \log (1-\varepsilon) - \delta_{\beta} \log \varepsilon - \varepsilon \log (1-\varepsilon) + \varepsilon \log \varepsilon
  \nonumber\\
  & = - \log (1 - \varepsilon) + \delta_{\beta} \log (1-\varepsilon) - \delta_{\beta} \log \varepsilon
  \nonumber\\
  & = - \log (1 - \varepsilon) + \delta_{\beta} \log (1-\varepsilon) - \delta_{\beta} \log \varepsilon
  \nonumber\\
  & = -\delta_{\beta}\log(\varepsilon)-(1-\delta_{\beta})\log(1-\varepsilon) \,,
  \label{Eq.Entropy_Der_1}
\end{align}
where $(a)$ holds by definition of a tangent line to a function at specific point, $(b)$ follows since derivative of the entropy function reads the negative of the logit function, i.e., $\diff{H(\delta_{\beta})}{\delta_{\beta}} = - \text{\rm logit}\,(\delta_{\beta}) \triangleq - \log \left( \frac{\delta_{\beta}}{1-\delta_{\beta}} \right)$ for $0 < \delta_{\beta} < 1$, and $(c)$ holds by definition of the entropy function, i.e.,
\begin{align}
    \label{Eq.Binary_Entropy}
    H(\varepsilon) \triangleq - \varepsilon \log \varepsilon - (1 - \varepsilon) \log (1 - \varepsilon) \,.
\end{align}
Therefore exploiting \eqref{Eq.Entropy_Der_1} we obtain,
\begin{align}
    \label{Eq.Entropy_Der_2}
    T_{\varepsilon}\left(\frac{k}{n}\right) = -\frac{k}{n}\log(\varepsilon)-(1-\frac{k}{n})\log(1-\varepsilon) \,,
\end{align}
which implies $-nT_{\varepsilon}(\frac{k}{n}) = k\log(\varepsilon) + (n-k) \log(1-\varepsilon)$. Thereby,
\begin{align}
    \label{Eq.Binom_1}
    2^{-nT_{\varepsilon}(\frac{k}{n})} = \varepsilon^k (1-\varepsilon)^{n-k} \,.
\end{align}
Now, observe that the Binomial coefficient $\binom{n}{k}$ where $k \geq 1$ and $n - k \geq 1$, can be upper bounded as follows \cite[see P.~353]{Cover91}
\begin{align}
    \label{Ineq.Binom_Coeff_UB}
    \binom{n}{k} \leq 2^{nH(\frac{k}{n})} \,.
\end{align}
Therefore,
\begin{align}
    \frac{k(1-\varepsilon)}{k(1-\varepsilon)-(n-k)\varepsilon} \cdot \binom{n}{k} \varepsilon^k (1-\varepsilon)^{n-k}
    & \stackrel{(a)}{\leq} \frac{k(1-\varepsilon)}{k(1-\varepsilon)-(n-k)\varepsilon} \cdot 2^{nH(\frac{k}{n})} \cdot \varepsilon^k (1-\varepsilon)^{n-k}
    \nonumber\\
    & \stackrel{(b)}{\leq} \frac{k(1-\varepsilon)}{k(1-\varepsilon)-(n-k)\varepsilon} \cdot 2^{nH(\frac{k}{n})} \cdot 2^{-nT_{\varepsilon}(\frac{k}{n})}
    \nonumber\\
    & = \left[ \frac{k(1-\varepsilon)}{k(1-\varepsilon)-(n-k)\varepsilon} \right] \cdot 2^{n\left[H(\frac{k}{n})-T_\varepsilon(\frac{k}{n})\right]}
    \,,
\end{align}
where $(a)$ holds by \eqref{Eq.Binom_1} and $(b)$ follows by exploiting \eqref{Eq.Binom_1}. Now, recalling \eqref{Ineq.Binom_Tail_2}, we obtain
\begin{align}
    \sum_{j=k}^n \binom{n}{j} \varepsilon^j (1-\varepsilon)^{n-j} \leq \frac{k(1-\varepsilon)}{k(1-\varepsilon)-(n-k)\varepsilon}  2^{n\left[H(\frac{k}{n})-T_\varepsilon(\frac{k}{n})\right]}
\end{align}
This completes the proof of Lemma~\ref{Lem.Bound_Sum_Tail_Binomial_Distribution_Part_2}.
\end{proof}
\end{lemma}
\section{Bound on The Binomial Cumulative Distribution Function}
\label{App.}
\begin{lemma}
\label{Lem.Binom_Tail_3}
Let $0 < \varepsilon < 1$ and $k < n$ with $\frac{k}{n} < \varepsilon$. Then,
\begin{align}
    \label{Ineq.Binom_Tail_3}
    \sum_{j=0}^k \binom{n}{j} \varepsilon^j(1-\varepsilon)^{n-j} \leq \frac{\varepsilon(n-k)}{\varepsilon n-k} \cdot
    2^{n\left[ H(\frac{k}{n})-T_{\varepsilon}(\frac{k}{n}) \right]} \,.
\end{align}
\end{lemma}
\begin{proof}
Let define
\begin{align}
    \label{Ineq.Exch_Variable_1}
    k' \triangleq n - k \,,
    \nonumber\\
    \varepsilon' \triangleq 1 - \varepsilon \,.
\end{align}
i.e., $k \leftrightarrow k'$ and $\varepsilon \leftrightarrow \varepsilon'$ or equivalently
\begin{align}
    \label{Ineq.Exch_Variable_2}
    k & \leftrightarrow n - k \,,
    \nonumber\\
    \varepsilon & \leftrightarrow 1 - \varepsilon \,.
\end{align}
Now, observe that
\begin{align}
    \label{Ineq.Exch_Variable_3}
    \frac{k}{n} > \varepsilon \Rightarrow \frac{k'}{n} < \varepsilon' \,.
\end{align}
Furthermore, by definition of the binary entropy function and its tangent line, we have
\begin{align}
    \label{Eq.Entropy_Exch_Variable}
    H\left( \frac{k}{n} \right) & = H \left( \frac{n-k}{n} \right) \,,
    \\
    \label{Eq.Entropy_Der_1_Exch_Variable}
    T_{\varepsilon} \left( \frac{k}{n} \right) & = T_{1 - \varepsilon} \left( \frac{n-k}{n} \right) \,.
\end{align}
where \eqref{Eq.Entropy_Exch_Variable} follows by \eqref{Eq.Binary_Entropy} and \eqref{Eq.Entropy_Der_1_Exch_Variable} holds by \eqref{Eq.Entropy_Der_2}.

Now applying the variable exchange of $j \leftrightarrow n-j$ unto \eqref{Ineq.Binom_Tail_2}, we obtain
\begin{align}
    \label{Ineq.Binom_Tail_4}
    \sum_{n-j=k}^{n-j = n} \binom{n}{n-j} \varepsilon^{n-j}(1-\varepsilon)^{n-(n-j)} \leq 2^{n\left[ H(\frac{k}{n})-T_{\varepsilon}(\frac{k}{n}) \right]} \left[ \frac{k(1-\varepsilon)}{k(1-\varepsilon)-(n-k)\varepsilon} \right] \,.
\end{align}
Observe that since the index of sum in \eqref{Ineq.Binom_Tail_2} runs form $k$ to $n$, i.e., $k \leq j \leq n$, in the new system we have $k \leq n - j \leq n$ which is equivalent to $0 \leq j \leq n - k$. Further, the Binomial coefficient satisfy
\begin{align}
    \binom{n}{n-j} = \binom{n}{j} \,,
\end{align}
for $0 \leq j \leq n$. Thereby,
\begin{align}
    \label{Ineq.Binom_Tail_5}
    \sum_{j=0}^{n-k} \binom{n}{j} \varepsilon^{n-j}(1-\varepsilon)^j \leq 2^{n\left[ H(\frac{k}{n})-T_{\varepsilon}(\frac{k}{n}) \right]} \left[ \frac{k(1-\varepsilon)}{k(1-\varepsilon)-(n-k)\varepsilon} \right] \,.
\end{align}
Now, applying the exchange of variables given in \eqref{Ineq.Exch_Variable_2} unto \eqref{Ineq.Binom_Tail_5}, we obtain
\begin{align}
    \label{Ineq.Binom_Tail_6}
    \sum_{j=0}^{k} \binom{n}{j} (1-\varepsilon)^{n-j}\varepsilon^j & \leq 2^{n[H(\frac{n-k}{n})-T_{1-\varepsilon}(\frac{n-k}{n})]} \left[ \frac{(n-k)\varepsilon}{(n-k)\varepsilon - k(1-\varepsilon)} \right]
    \nonumber\\
    & = 2^{n\left[ H(\frac{k}{n})-T_{\varepsilon}(\frac{k}{n}) \right]} \left[ \frac{(n-k)\varepsilon}{(n-k)\varepsilon - k(1-\varepsilon)} \right] \,,
\end{align}
where the equality holds by \eqref{Eq.Entropy_Exch_Variable} and \eqref{Eq.Entropy_Der_1_Exch_Variable}. Therefore,
\begin{align}
    \label{Ineq.Binom_Tail_7}
    \sum_{j=0}^{k} \binom{n}{j} (1-\varepsilon)^{n-j}\varepsilon^j \leq 2^{n\left[ H(\frac{k}{n})-T_{\varepsilon}(\frac{k}{n}) \right]} \left[ \frac{(n-k)\varepsilon}{(n-k)\varepsilon - k(1-\varepsilon)} \right] \,.
\end{align}
Now we focus on the bracket in \eqref{Ineq.Binom_Tail_6} which can be simplified as follows
\begin{align}
    \frac{(n-k)\varepsilon}{(n-k) \varepsilon - k (1 - \varepsilon)} & = \frac{\left( \frac{n-k}{n} \right) \varepsilon}{\left( \frac{n-k}{n} \right) \varepsilon - \frac{k}{n} (1 - \varepsilon)}
    \nonumber\\
    & = \frac{\varepsilon-\frac{k}{n}\varepsilon}{\varepsilon - \frac{k}{n}}
    \nonumber\\
    & = \frac{\varepsilon(n-k)}{\varepsilon n-k} \,,
\end{align}
where the first equality follows by dividing both sides in the left side by factor $n$. Thereby,
\begin{align}
    \sum_{j=0}^{k} \binom{n}{j} (1-\varepsilon)^{n-j}\varepsilon^j \leq \frac{\varepsilon(n-k)}{\varepsilon n-k} \cdot
    2^{n\left[ H(\frac{k}{n})-T_{\varepsilon}(\frac{k}{n}) \right]} \,.
\end{align}
This completes the proof of Lemma~\ref{Lem.Binom_Tail_3}.
\end{proof}
\section*{}
\bibliography{Lit}
\end{document}